\documentclass[journal]{IEEEtran}

\usepackage{graphicx}
\usepackage[cmex10]{amsmath}
\usepackage{amssymb,amsthm}
\usepackage{algorithm, algorithmic}
\newtheorem{theorem}{Theorem}[section]
\newtheorem*{theorem*}{Theorem} 
\newtheorem{lemma}[theorem]{Lemma}
\newtheorem*{lemma*}{Lemma} 

\newtheorem{corollary}[theorem]{Corollary}

\newcommand{\be}{\begin{equation}}
\newcommand{\ee}{\end{equation}}
\newcommand{\F}{\mathcal{F}}
\newcommand{\R}{\mathbb{R}}
\newcommand{\Q}{\mathbb{Q}}
\newcommand{\Z}{\mathbb{Z}}
\newcommand{\T}{\mathbb{T}}

\newcommand{\st}{\text{s.t.}}

\newcommand{\prob}{\text{Prob}}

\newcommand{\one}{\textbf{1}}

\begin{document}

\title{Reconstruction of Binary Functions and Shapes from Incomplete Frequency Information}

\author{Yu Mao

Institute for Mathematics and Its Applications, University of Minnesota.}

\maketitle

\begin{abstract}
The characterization of a binary function by partial frequency information is considered. We show that it is possible to reconstruct binary signals from incomplete frequency measurements via the solution of a simple linear optimization problem. We further prove that if a binary function is spatially structured (e.g.  a general black-white image or an indicator function of a shape), then it can be recovered from very few low frequency measurements in general. These results would lead to efficient methods of sensing, characterizing and recovering a binary signal or a shape as well as other applications like deconvolution of binary functions blurred by a low-pass filter. Numerical results are provided to demonstrate the theoretical arguments.
\end{abstract}

\section{Introduction}

This paper discusses the reconstruction of the binary signals. Binary signals appear in a variety of applications like shape processing, bar code and handwriting recognition, obstacle detection, image segmentation; see e.g.  \cite{Litman:1998p7217,Osher:2001p7372,Ishikawa:2003p7731, Esedoglu:2004p7220, Alajlan:2008p7326,  Pock:2010p7690} and many others. 

One of the major difficulties in the reconstruction of binary functions is that the binary constraint is non-convex.  Optimization with a binary constraint is often approached by means of the double-well potential or other nonlinear schemes. In this paper we demonstrate that binary functions can be reconstructed exactly via a simple convex optimization when only partial frequency information is available (e.g.  when the signal is blurred by a low-pass filter). 

\subsection{Main results}

Let $u_0$ be a binary function, i.e. $u_0(x)\in\{0,1\}$, $\forall x$. Let $\F$ be the Fourier transform and $S$ the selecting operator corresponding to the incomplete  measurements $b$.  Our goal is to recover $u_0$ from $b=S\F u_0$, which is an underdetermined problem. The main contribution of this work is showing that under certain conditions, $u_0$ can be exactly reconstructed by solving the convex relaxed optimization problem
\[
\text{find}\quad u \quad \st\quad S \F u=b,\quad  0\leq u\leq 1.
\]
At first glance, this may seem a little bit surprising, as it is not even obvious that the solution of the problem is unique. However, in this work we prove that in many cases, the solution is unique and is equal to $u_0$.
\begin{itemize}
\item When a binary signal is spatially structured, i.e. the 1s and 0s are clustered (e.g.  in a binary image or as an indicator function of a shape), with very few low frequency measurements taken, the solution of this convex optimization problem is deterministically unique and equals to the original binary signal. For a detailed statement, see theorem \ref{lowfreq1D}.
\item If a binary signal has no spatial structure, for example if the 1s and 0s appear randomly, we show that this relaxation works with overwhelming probability when the number of the measurements is more than a half of the size of the signal, and the probability tends to 1 as the size of the signal increases to infinity. For a detailed statement, see theorem \ref{randombinary}.
\item We also propose a very efficient algorithm designed for this convex problem (see algorithm \ref{OverallAlgorithm}). Numerical experiments are presented in section \ref{numerical}. 
\end{itemize}

\subsection{Related works}

The idea that under certain circumstances, the binary constraint can be automatically satisfied by imposing a convex relaxation, in particular the box constraint $0\leq u\leq 1$, is not new. For example, when solving the image segmentation, multi-label and many other problems based on the total variation model (see e.g.  \cite{Chan:2006p2646, Bresson:2007p7744, Pock:2010p7690}), people have noticed that although the original problem is non-convex, the global minimizer can be obtained by solving the relaxed convex problem and the solution will be automatically (almost) binary. However, this approach works only because of the special structure of the variational model. The theoretical analysis strongly depends on the coarea formula for total variation. 

In the contexts of regression and approximation, it is has been known for a long time that in an $L_\infty$ regression (sometimes referred to as Chebyshev or minimax regression where the penalty function is given by the $L_\infty$ norm) or a deadzone-linear penalty regression (where the penalty function is given by the deadzone function $(|\cdot|-a)^+$), the distribution of the residual of the regression will concentrate at the boundary of the feasible domain (see e.g.  \cite{Boyd:2004p1827}, Chapter 6). When the feasible domain is interval $[0,1]$, a function with many values right at the boundary is nothing but a binary signal. In fact, in section \ref{theorygen} we will show that our convex relaxation of the problem can be equivalently reformulated as an $L_\infty$ or deadzone penalty minimization problem. 

The idea of recovering a signal from the partial frequency measurements is often used for compressed sensing \cite{Candes:2006p625, Candes:2006p1807, Donoho:2006p668}, which takes advantage of the prior assumption on the sparsity of the signal and reconstructs the signal via $L_1$ minimization. However, the current work is substantially different from compressed sensing. Although structured binary functions are  a special case of piecewise constant functions whose derivative is sparse, the condition of being binary is actually stronger than simply being piecewise constant, therefore stronger results can be expected. Indeed, none of the results given in this work can be deduced from the standard compressed sensing theory directly, and some of them are of a very different nature. For example, in compressed sensing the frequency measurements should be taken randomly to guarantee the restricted isometry property \cite{Candes:2006p625}, while in the reconstruction of the structured binary function, the low frequency measurements actually play a more important role than the high frequency measurements as discussed below. On the other hand, many major results given in this paper are deterministic, while results in the compressed sensing literature are often intrinsically stochastic. 

The present research is also related to a seminal work on the reconstruction of signals from partial frequency information \cite{Donoho:1989p3875}, where the spatial structures and patterns in both time and frequency domain are used to guarantee the uniqueness of the signal reconstruction. In the very recent research \cite{Donoho:2010p5288,Donoho:2010p6951}, the authors showed that a random binary signal can be recovered with certain probability by means of the relaxed box constraint. In their work the major mathematical tool is the delicate geometric face-counting of random polytopes. We also get a basically similar result in section \ref{theoryrandom}, but from a different approach.  In \cite{Vetterli:2002p7065}, the authors defined the degrees of freedom contained in a sparse or piecewise polynomial signal as the \emph{rate of innovation} of the signal. Then they showed that the quantity of the  samplings needed to recover the signal equals to the rate of innovation. However, the mathematics behind their theory is substantially different than ours. Moreover, their method requires a certain pattern of sampling and the reconstruction involves a factorization of polynomial. In \cite{Bruckstein:2008p2703} the authors showed that if an underdetermined system admits a very sparse nonnegative solution and the matrix has a row-span intersecting the positive orthant, the solution is actually unique. In \cite{Fuchs:2005p8032} the author proved that a sparse nonnegative can be reconstructed as the unique solution of a linear programming problem, where the corresponding matrix is the submatrix of a Fourier matrix consisting of its top rows. We will further discuss the relationship between these two works and ours in section \ref{discussion}.
\subsection{Notations and conventions}

In this paper all signals are assumed to have periodic boundary condition. The $h$-dimensional discrete signals are defined on $\{1,\ldots,N\}^h$ where $N$ is assumed to be even. Here for simplicity we assume that the domain is equilong along each dimension. The $h$-dimensional continuous signals are defined on $\T^h=[0,1]^h$ where the two endpoints $0$ and $1$ are identified due to the periodic boundary condition. 

We use $\F$ to denote the Fourier transform, both in the discrete and continuous periodic cases. 

When we talk about the discrete Fourier transform, we use the following convention:
\[
a_k=\sum_{x\in[1,N]^h} u(x)e^{-2\pi i \langle k,\frac{x}{N}\rangle}
\]
\[
u(x)=\frac{1}{N^h}\sum_{k\in[-\frac{N}{2},\frac{N}{2}-1]^h} a_ke^{2\pi i \langle k,\frac{x}{N}\rangle}
\]
where $\{a_k\}$ are the Fourier coefficients defined on a symmetric support ($N/2$ is treated as same as $-N/2$). A smaller $|k|$ corresponds to a lower frequency. Since $u$ is always real, $\{a_k\}$ satisfies $a_{-k}=\overline{a_k}$. Notice that $\F^\top=N^h\F^{-1}$.  

We use $S$ to denote the selecting operator. $S$ is a diagonal matrix where the selected positions have value $1$ and others are $0$.

\subsection{Contents}
The paper is organized as follows. Section \ref{theory} discusses the theoretical results. In section \ref{solving} an algorithm to solve the convex problem is proposed. Numerical experiments are shown in section \ref{numerical} and conclusion is given in section \ref{conclusion}. To make the main text more concise, we put all proofs into Appendix except those theorems and corollaries immediately deduced from the discussion in the context. 

\section{Theory}\label{theory}

\subsection{General reconstruction theory}\label{theorygen}

Suppose $u_0$ is a discrete binary signal, i.e. $u_0(x)\in\{0,1\},\forall x$. Consider a linear system $Au_0=b$ where $A=S \F$, $\F$ is the Fourier transform and $S$ is the selecting operator. The meaning of this system is clear: some partial frequency information of the binary signal is given, and we want to reconstruct $u_0$ from the incomplete measurements. This leads to the following problem ($P_0$):
\be\label{P_0}
P_0:\quad\text{find}\quad u \quad \st\quad Au=b,\quad  u(x)\in\{0,1\}.
\ee
The problem $(P_0)$ is non-convex due to the binary condition, and the following convex problem is the tight relaxation of $(P_0)$:
\be\label{P_1}
P_1:\quad\text{find}\quad u \quad \st\quad Au=b,\quad  0\leq u\leq 1.
\ee
We want to show that $(P_1)$ can be used to recover $u_0$ under certain conditions. The following theorem specifies the conditions guaranteeing that this relaxation is exact.

\begin{theorem}\label{theorem1}
Assume $u_0$ is a binary solution of $Au_0=b$. There exists no nonzero $v\in\{Av=0\}$ such that 
\be\label{orthant}
\begin{cases}v(x)\leq0, &\text{when}\quad u_0(x)=1\\
v(x)\geq0, &\text{when}\quad u_0(x)=0
\end{cases}
\ee
if and only if $u_0$ is the unique solution of $(P_1)$, i.e. solving $(P_1)$ recovers $u_0$.
\end{theorem}

If the size of the signal is $N$, then the criteria \eqref{orthant} on $v$ determines an orthant in $\R^N$ depending on $u_0$. We denote this orthant as $\mathbb{O}_{u_0}$. The theorem tells us that as long as the kernel space of $A$ intersects $\mathbb{O}_{u_0}$ at nowhere but the origin, solving $(P_1)$ is enough to recover $u_0$. 

This condition is `negative', i.e. it requires the nonexistence of such a vector $v$. The following statement, sometimes referred to as the Gordan-Stiemke theorem of the alternative, will lead to a `positive' criteria. 

\begin{lemma}[Alternative Theorem, see e.g.  \cite{Boyd:2004p1827}]\label{alternative}
One and only one of the two following problems is feasible: (1). Find $0\neq v\geq 0$ s.t. $Av=0$; (2). Find $v=A^\top \eta$ s.t. $v>0$. 
\end{lemma}

Geometrically, this theorem says that if $P$ is a subspace in $\R^N$, $\mathbb{O}$ is the first orthant, then either $P\bigcap \mathbb{O}=\{0\}$ or $P^\perp\bigcap \text{int}(\mathbb{O})=\emptyset$, but not both. The statement considers the first orthant only, but obviously it is true for any other given orthant. Apply this lemma on theorem \ref{theorem1}, we immediately get the `positive' version of the criteria:

\begin{theorem}\label{theorem1p}
Assume $u_0$ is a binary solution of $Au_0=b$.There exists $v=A^\top\eta$ such that 
\be
\begin{cases}v(x)<0, &\text{when}\quad u_0(x)=1\\
v(x)>0, &\text{when}\quad u_0(x)=0
\end{cases}
\ee
if and only if $u_0$ is the unique solution of $(P_1)$, i.e. solving $(P_1)$ recovers $u_0$.
\end{theorem}

Unfortunately, there is no explicit formula to determine if an arbitrary subspace passes through a given orthant. Indeed, it is equivalent with any general linear programming feasibility problem and thus has no closed-form solution. However, for some special cases we can still give deterministic or stochastic results, as we will explain in the following subsections.

We want to remark that there are many alternative linear programming problems that can recover the signal as well. In fact, assume $J(u)$ is a convex function on $u$ satisfying  
\be\label{P_2condition}
J(u)<J(v),\forall u\in[0,1]^N, v\notin[0,1]^N.
\ee
It is easy to see that if $u_0$ is a binary solution of $Au_0=b$, then $u_0$ is a unique solution of $(P_1)$ implies that $u_0$ is a unique solution of the following convex problem:
\be\label{P_2}
\quad\min_u J(u) \quad \st\quad Au=b.
\ee
Therefore solving \eqref{P_2} can also recover $u_0$ under the condition in theorem \ref{theorem1} or \ref{theorem1p}. There are many simple functions satisfying \eqref{P_2condition}. One of the simplest examples is $J(u)=\|2u-1\|_\infty$. Another example is the deadzone penalty $J(u)=\tilde J((|2u-1|-1)^+)$ where $\tilde J$ is any convex function with $\tilde J(0)=0$ and $\tilde J(v)> 0$ for $v\neq 0$, e.g.  $\tilde J(v)=\|v\|_p$ for $p\geq 1$. 

\subsection{Reconstruction of the 1D binary signals from the low frequency measurements}\label{theory1D}

So far the discussion has used only the fact that the signal to be reconstructed is binary. In most practical applications, the signal is often not only binary, but also structured, i.e. the 1s and 0s are spatially clustered. This property could help us reconstruct the signal.

Let us consider the 1D case first. Assume $u_0(x)$ is a periodic discrete binary signal defined on $\{1,\ldots,N\}$. Since we are considering the structured signal, $u_0(x)$ consists of many intervals with constant value $1$ or $0$. If the first and last intervals are with the same value, we treat them as one merged interval under the periodic boundary condition. Therefore, the total number of intervals is always even, thus $u_0(x)$ can be represented as 
\be\label{blockwise}
u_0=\sum_{j=1}^{2d} \xi_j\one_{I_j},\quad \xi_j\in\{0,1\}.
\ee
$\{I_j\}$ is a partition of $\{1,\ldots,N\}$ where each $I_j$ is a consecutive interval.

From theorem \ref{theorem1p}, $u_0$ can be recovered from $(P_1)$  if and only if there exists $v=A^\top\eta$ such that
\be\label{blockfeasibility}
\begin{cases}v(x)<0\ \text{in}\ I_j & \text{if}\ \xi_j=1\\
v(x)>0\ \text{in}\ I_j & \text{if}\ \xi_j=0
\end{cases}
\ee
We want to show that this condition is always satisfied for certain types of $A$. Recall that when partial frequency information is given, $A=S\F$ where $S$ is a sampling operator that corresponds to the known frequencies. If $v=A^\top\eta=\F^\top(S\eta)$, then $v$ is a band-limit signal whose spectrum can be represented as $S\eta$, i.e. it is located inside the known frequencies.  Therefore, the above condition means that the relaxation method is valid as long as we can use only those known frequencies to construct a band-limit signal that satisfies \eqref{blockfeasibility}. Since \eqref{blockfeasibility} describes the zero-crossing position of $v$, it imposes a constraint on the spectrum of $v$, and therefore on $S$. 

The relationship between the zero-crossings of a signal and its spectrum information is not a new problem in signal processing; readers are referred to \cite{Logan:1977p7105, Requicha:1980p7107, Kedem:1986p7053, Ulanovskii:2006p7050} for some classic theories. The following result is natural from the perspective of trigonometric interpolation:

\begin{lemma}\label{trigonointerp}
Let $\T=[0,1]$ where $0$ and $1$ are identified, i.e. $\T\cong S^1=\{z:|z|=1\}$. Given $2n$ points on $\T$ who define $2n$ intervals on $\T$, there exists a real trigonometric polynomial, whose spectrum is limited in $[-n,n]$, vanishing only at those points and changes signs alternatively on those intervals.
\end{lemma}

This conclusion, combined with theorem \ref{theorem1p}, leads to the following deterministic result which states that the number of low frequency measurements we need to reconstruct the binary function is basically the number of the jumps contained in the signal, no matter how large the signal is. 

\begin{theorem}\label{lowfreq1D}
If $u_0(x)$ is a 1-D binary signal that can be represented as in \eqref{blockwise} with $2d$ consecutive intervals of ones and zeros, then by knowing the Fourier coefficients $\{a_k\}$ for $|k|\leq d$, we can recover $u_0$ through the convex problem $(P_1)$. (Notice that $u_0(x)\in\R, \forall x$ implies $a_k=\overline{a_{-k}}, \forall k$, so essentially we only need to know $\{a_k\}$ for $0\leq k\leq d$.) This result is optimal, i.e. precise reconstruction via solving $(P_1)$ is impossible if knowing even less.
\end{theorem}

Although theorem \ref{lowfreq1D}  only holds when the lowest frequency information is given, it is still very useful, because in many practical problems the low frequency measurements are far easier to obtain than the high frequency measurements. The deconvolution problem with a low-pass filter kernel, for example, can be treated as reconstruction from the lowest frequency information. 

Heuristically, theorem \ref{lowfreq1D} can be understood as follows: if the low frequency measurements are given, then the permitted perturbation can be with higher frequencies only and thus strongly oscillating around zero. Therefore, by controlling the lower and upper bounds of the signal as in $(P_1)$, the oscillating perturbation would be eliminated, and thus the solution is uniquely determined.

\subsection{Discussions and generalizations}\label{discussion}

First, it is easy to see that theorem \ref{lowfreq1D} can be directly extended to the cosine transform as well, due to the fact that the cosine transform of a signal is nothing but the Fourier transform of the even extension of the signal. 

\begin{corollary}\label{lowfreqDCT1D}
If $u_0(x)$ is a 1-D binary signal that can be represented as in \eqref{blockwise} with $2d$ consecutive intervals of ones and zeros, then by knowing the discrete cosine transform coefficients $\{a_k\}$ for $0\leq k\leq 2d$, we can recover $u_0$ through the convex relaxation $(P_1)$. 
\end{corollary}

We also mention that if the signal is only bounded from one side, i.e. instead of knowing that the signal is binary, we know the signal is nonnegative, then a similar argument would lead to a theorem concerning the reconstruction of the sparse nonnegative signals. Indeed, using theorem \ref{theorem1} and lemma \ref{alternative} we can obtain the following theorem (the proof is similar hence omitted):

\begin{theorem}\label{sparse1D}
If $u_0\geq 0$ is supported on $K=\{x:u_0(x)\neq 0\}$, then $u_0$ is the unique solution of $Au=b$, $u\geq 0$ if and only if there exists $v=A^\top \eta$ such that $v|_K=0$, $v|_{K^c}>0$. 
\end{theorem}

Let $A=S\F$ and $S$ also select the low frequency measurements, then theorem \ref{sparse1D} implies the following theorem (thanks to lemma \ref{trigonointerp} as well):

\begin{theorem}\label{lowfreqsparse1D}
If $u_0(x)$ is a 1-D nonnegative sparse signal supported on $K=\{x:u_0(x)=0\}$ with $|K|=d$, then by knowing the Fourier coefficients $\{a_k\}$ for $|k|\leq d$, we can recover $u_0$ through the convex problem $Au=b$, $u\geq 0$. 
\end{theorem}

This result is closely related with the theorems proved in \cite{Bruckstein:2008p2703} which said that a nonnegative solution of a linear system is unique if the solution is sparse enough and the matrix has a row-span intersecting the positive orthant. It is worth mentioning that the quantity of the needed low frequency measurements in this case approximately equals two times the quantity of the `spikes' of $u_0$ (other than the number of jumps of $u_0$ in theorem \ref{lowfreq1D}). This coincides the observation in \cite{Vetterli:2002p7065} that the degree of freedom of a $d$-sparse signal is $2d$ (for each spike there is one degree for position and one for amplitude), and thus $2d+1$ measurements are in principal enough. A similar observation has been given in \cite{Fuchs:2005p8032} as well.

Our result may further be generalized to bases other than the trigonometric functions. Indeed, the duality of lemma \ref{trigonointerp} tells us that a signal without lower frequency components must have many sign changes, which is  some times referred to as the Sturm-Hurwitz theorem \cite{Ulanovskii:2006p7050}. This observation plays a critical role here. This property can be extended to other basis that has similar oscillating pattern, such as some wavelet bases or the eigenfunctions of the regular Sturm-Liouville problems \cite{Galaktionov:2005p7108}. However, generalization along this line is beyond the scope of this paper.

\subsection{Reconstruction of the 2D binary signals from the low frequency measurements}\label{theory2D}

The multidimensional case is more complicated than the 1D case due to the following several reasons. There is no fundamental algebraic theorem for multivariable polynomials.  Moreover, the Sturm-Hurwitz theorem that describes the zero-crossings of function with a spectrum gap does not exist in higher dimensions. Finally,  in 1D the complexity of a binary function can be simply characterized by the number of jumps as in theorem \ref{lowfreq1D}, while in higher dimensions, a binary function may have only one connected component but still have a very complicated jump set. 

Since theorem \ref{theorem1p} is still valid in the multidimensional case, a multidimensional binary signal can be reconstructed by the lower frequency measurements as long as the jump set of the binary signal is the zero levelset of a low frequency function. Unfortunately, to the author's knowledge, no criteria has been known to determine if a given shape can be realized as the zero levelset of a function with only lower frequency components. In \cite{Curtis:1985p7206,CURTIS:1987p7208,Rotem:1986p7178,Sanz:1989p7201,Sanz:1989p7122,Hummel:1989p7209,Zakhor:1990p7196}, some results concerning the relationship between the levelset of a function and its Fourier transform are shown. In \cite{Nashold:1989p7110} it has been proved that using the continuous Fourier transform, a function with a given levelset curve can be approximated to any degree of accuracy by a band-limited function with given spectrum support. However, this result is barely useful in practice because it requires virtually infinitely high resolution in the frequency domain. 

Heuristically, if a function has only lower frequency components, we can imagine that its levelset would not be too complicated. Here we give a way to measure this complexity.  The basic idea is that since in 1D case the complexity of a binary signal is determined by the number of jumps inside the signal, in 2D case we can define an `average number of zero-crossings' as illustrated in Fig. \ref{grating}. The following discussion can be naturally extended to higher dimensional cases.

\begin{figure}[b]
\begin{center}
\includegraphics[width=1.5in]{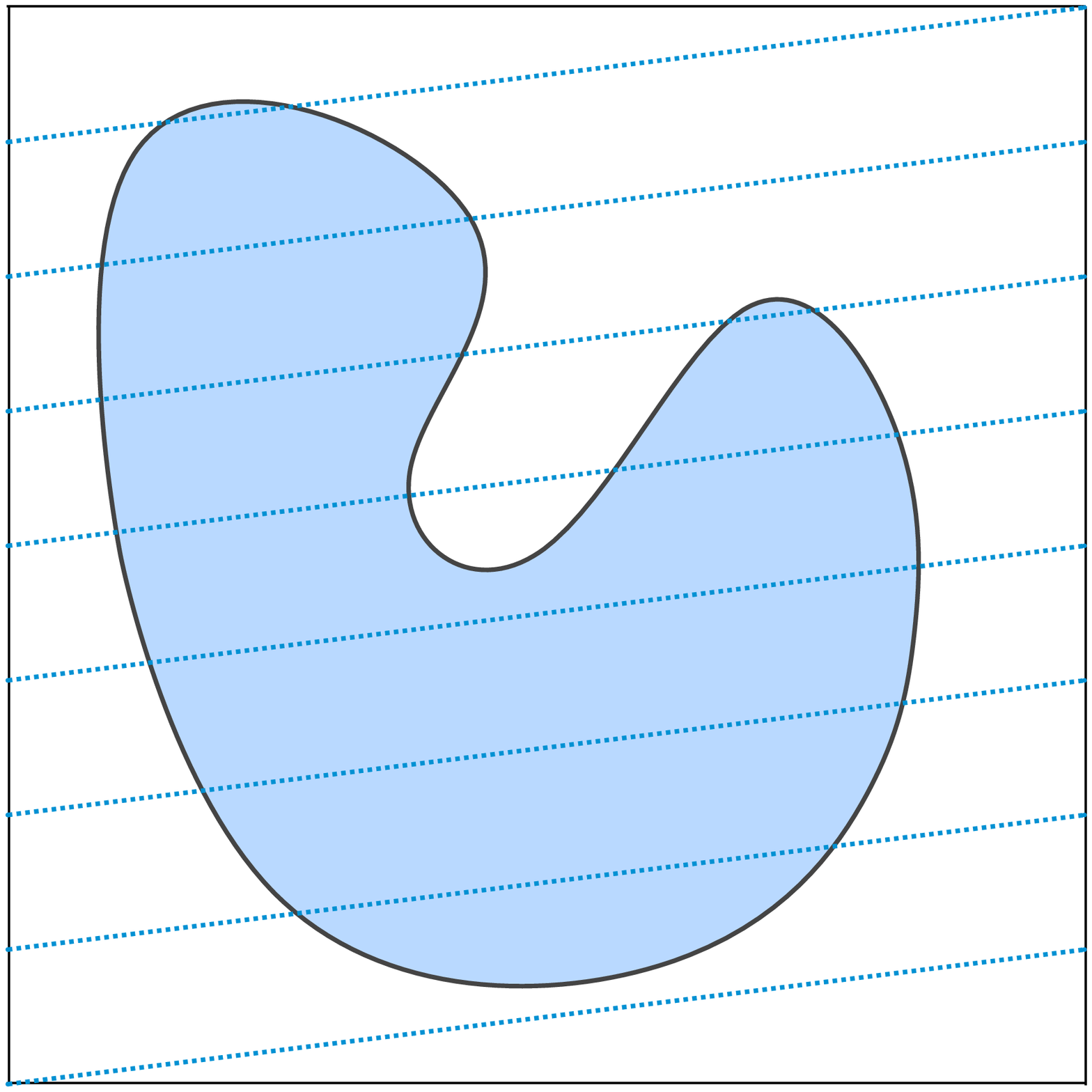}
\includegraphics[width=1.5in]{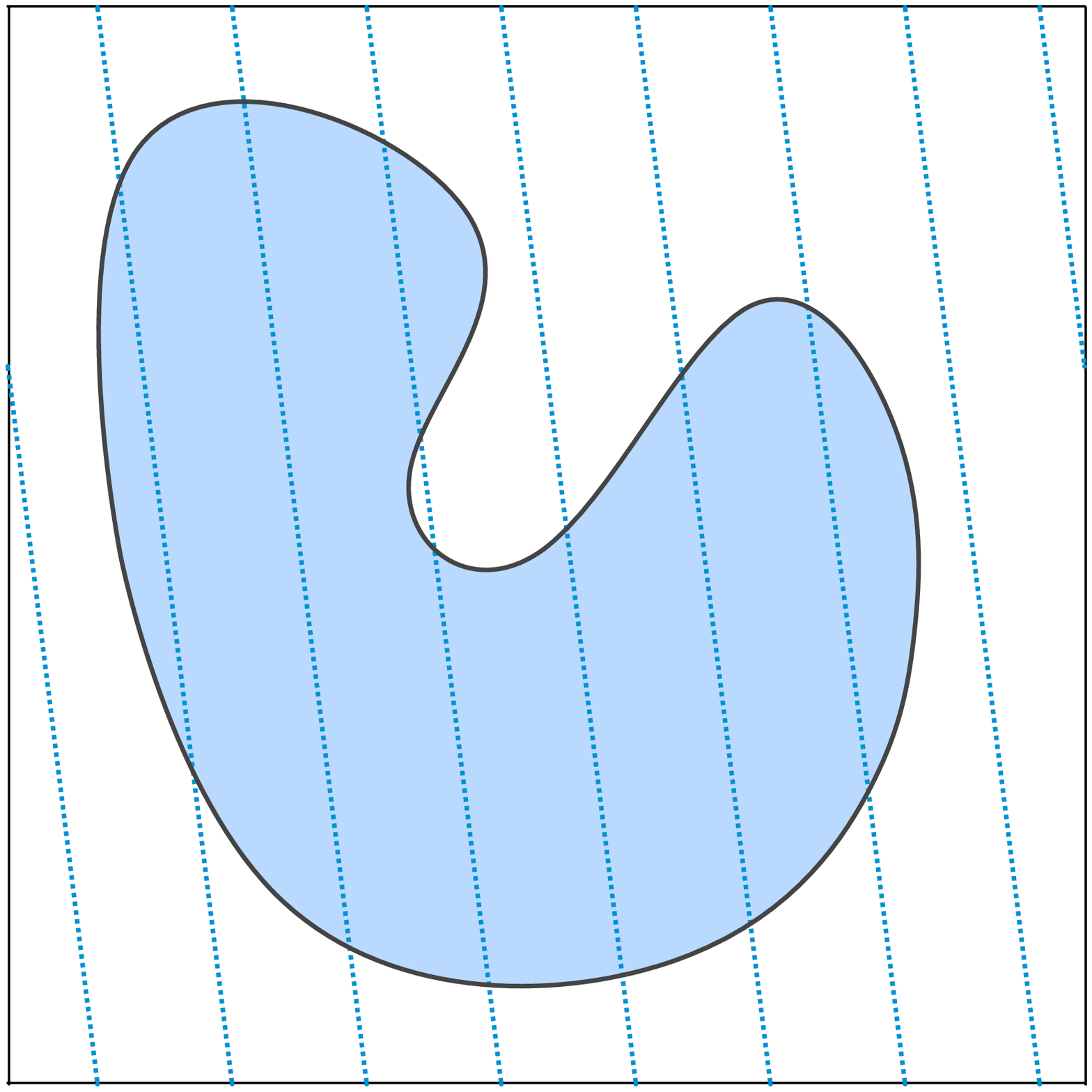}
\end{center}
\caption{Two gratings on $\T^2$ with a binary function.}\label{grating}
\end{figure}

Define $\T^2=[0,1]\times[0,1]$ with opposite boundaries identified, that is,  $\T^2\cong\R^2/\Z^2$. For  $\theta\in(-\pi/4,\pi/4]$, define 
\[
L_{s,\theta}(t)=(t, s+t\tan\theta)\mod 1, \quad s\in[0,1], t\in[0,1]
\]
to be a \emph{grating} along angle $\theta$ starting from the left edge of the square (see Fig. \ref{grating} left). For  $\theta\in(\pi/4,3\pi/4]$, similarly define 
\[
L_{s,\theta}(t)=(s+t\cot\theta, t)\mod 1, \quad s\in[0,1], t\in[0,1]
\]
to be the grating along angle $\theta$ starting from the bottom edge of the square (see Fig. \ref{grating} right). Assume $u$ is a binary function whose jump set consists of analytic curves. For given $\theta$ and $s$ the line segment $L_{s,\theta}$ will intersect the jump set of $u$ finite times. We denote this number as $\#L_{s,\theta}$ and define the average of $\#L_{s,\theta}$ over $s$ as
\[
K_\theta=\cos\theta\int_0^1 \#L_{s,\theta}ds
\]
for $\theta\in(-\pi/4,\pi/4]$ and
\[
K_\theta=\sin\theta\int_0^1 \#L_{s,\theta}ds
\]
for $\theta\in(\pi/4,3\pi/4]$. The presence of the multiplier $\cos\theta$ and $\sin\theta$ is due to the fact that for gratings with different angle $\theta$, $s$ is not an equilong variable. The distance between $L_{s,\theta}$ and the origin is more intrinsic  which equals $s\cos\theta$ for $\theta\in(-\pi/4,\pi/4]$ and $s\sin\theta$ for $\theta\in(\pi/4,3\pi/4]$. $K_\theta$ is called \emph{the average directional number of zero-crossings} in this paper which essentially describes the average quantity of sign changes along the direction $\theta$. It is easy to see the connection of this quantity and the number of jumps for a 1D binary signal. Indeed, by the Cauchy-Crofton formula, $\int K_\theta d\theta$ is nothing but two times the perimeter of the shape, i.e. the total variation of $u_0$, while in the 1D case the number of jumps also equals to the total variation of the  binary signal. In the next theorem we will show that $K_\theta$ in some sense characterizes the complexity of the shape.

\begin{theorem}\label{levelsetcondition}
Assume $u(x,y)$ is a 2D binary function with analytic jump curve and the average directional number of zero-crossings of $u$ along the angle $\theta$ is denoted as  $K_\theta$. If there exists a band-limited real function $v(x,y)=\sum_{(j,k)\in\Omega} a_{jk}e^{2\pi i(jx+ky)}$ defined on $\T^2$, where $\Omega=\{(j,k):\sqrt{j^2+k^2}\leq d\}$, such that the jump set of $u(x,y)$ corresponds to the zero levelset of $v(x,y)$, then  $K_\theta\leq 2d,\forall \theta$.
\end{theorem}

\begin{corollary}\label{lowfreq2D}
Assume $u_0(x,y)$ is a discrete 2D binary function defined on $\{1,\ldots,N\}^2$ and $u(x,y)$ is a binary function  defined on the continuous domain $\T^2$ with analytic jump curves such that $u(x/N,y/N)=u_0(x,y)$ for $(x,y)\in\{1,\ldots,N\}^2$, and denote the average directional number of zero-crossings of $u$ along the angle $\theta$ by $K_\theta$. If the reconstruction of $u_0$ by linear programming problem $(P_1)$ from low frequency measurements in $\Omega=\{(j,k):\sqrt{j^2+k^2}\leq d\}$ is exact, then $d\geq\frac{1}{2}\max_\theta K_\theta$.
\end{corollary}

The meaning of theorem \ref{levelsetcondition} and corollary \ref{lowfreq2D} is clear: the average directional number of zero-crossings of the levelset of a band-limited function is bounded by the diameter of the support of the spectrum. If we denote the jump set of $u(x,y)$ as $\Gamma$ and the perimeter $|\Gamma|$, then by Cauchy-Crofton formula, the condition $d\geq\frac{1}{2}\max_\theta K_\theta$ further implies $d\geq\frac{1}{\pi}|\Gamma|$, which can be seen as a natural generalization of the 1D case (see lemma \ref{trigonointerp} and theorem \ref{lowfreq1D}). However, unlike the 1D case, this theorem just gives the necessary condition, not a sufficient one.

\subsection{Reconstruction of binary signal from arbitrary frequency measurements}\label{theoryarbifreq}

If $S$ is an arbitrary frequency selector,  not necessarily selecting the lowest frequencies, it is not easy to give a sufficient and necessary condition to determine if the reconstruction is possible since there is no way to quantify the zero-crossings simply from the irregular support of the spectrum. In \cite{Kozma:2002p7049} the authors show that given the support of the spectrum of a trigonometric polynomial, the size of the largest non-zero circular region of the polynomial is bounded. They proved the following theorem:

\begin{theorem}\label{arbitraryfreqthm}
Let $0\notin S\subset \Z^d$ be a finite set s.t. $S=-S$. Let $v(x)=\sum_{k\in S}c_ke^{2\pi i\langle k,x\rangle}$ be a real valued trigonometric polynomial on $\T^d$, then $v(x)$ has at least one zero in any closed ball of diameter $\sum_{k\in S}\frac{1}{4\|k\|}$.
\end{theorem}

This theorem indicated that if the known frequencies have an arbitrary support, then the binary functions can be recovered from $(P_1)$ if it contains a constant block large enough. However, the bound given in theorem \ref{arbitraryfreqthm} is rather loose. 

It is worth mentioning that the conclusion of theorem \ref{arbitraryfreqthm} tells us the `importance' of each frequency band is roughly determined by the reciprocal of the frequency. That is to say, knowing lower frequency measurements is more important for reconstruction of the binary signals than knowing the high frequency measurements. This coincides with the intuition we learn from theorem \ref{lowfreq1D} and differs from the case of sparse reconstruction as in the compressed sensing problems, where the measurements should be spread out in the frequency domain as much as possible.

\subsection{Reconstruction of random binary signal}\label{theoryrandom}

If a binary function is random, i.e. the orthant $\mathbb{O}_{u_0}$ is randomly chosen, there is no deterministic way to guarantee if a certain subspace passes through it, but the probability can be estimated. From now on we denote the kernel of $A$, the image of $A^\top$ by $I_A$ and  the rank of $A$ by $K_A$, $I_A$ and $r$ respectively, then $\dim(K_A)=N-r$, $\dim(I_A)=r$. We say an $r$-dimensional linear subspace is \emph{in general position} if the projections of any $r$ axes of $\R^N$ onto the subspace are linearly independent, and we say $A$ is \emph{in general position} if $K_A$ is in general position. The following result has been known by mathematicians at least as far back as the 1950s (see \cite{Cover:1965p7055} for a brief review).   It says that any $r$-dimensional subspace in $\R^N$ in general position will pass through a fixed number of orthants of $\R^N$:

\begin{lemma}[see e.g.  \cite{Cover:1965p7055}]\label{counting}
Any $r$-dimensional subspace in $\R^N$ in general position passes through $2\sum_{i=0}^{r-1}{{N-1}\choose{i}}$ orthants of $\R^N$.
\end{lemma}

We denote $P_{r,N}=\sum_{i=0}^{r}{{N}\choose{i}}/2^N$, which is nothing but the cumulative distribution of the function of the standard binomial distribution with $p=\frac{1}{2}$. For a random binary signal $u_0$, we say $u_0$ has no `preference' on orthants if for any two orthants $\mathbb{O}_1$ and $\mathbb{O}_2$,
\be\label{orthantsymmetric}
\prob(\mathbb{O}_{u_0}=\pm\mathbb{O}_1)=\prob(\mathbb{O}_{u_0}=\pm\mathbb{O}_2),
\ee
then since there are $2^N$ orthants in total, we have
\be\label{binomialprob}
P\left(I_A\bigcap\text{int}\left(\mathbb{O}_{u_0}\right)\neq\emptyset\right)=\frac{2\sum_{i=0}^{r-1}{{N-1}\choose{i}}}{2^N}=P_{r-1,N-1}
\ee 
According to theorem \ref{theorem1p}, this is equivalent to saying:
\begin{theorem}
If $u_0$ is a random binary signal with size $N$ without preference on orthants, then given a matrix $A$ in general position with rank $r$, the probability that $u_0$ can be recovered from linear problem $(P_1)$ is $P_{r-1,N-1}$. 
\end{theorem}

It is well known that $P_{r,N}$ can be approximated by $\Phi\left(\frac{2r-N}{\sqrt{N}}\right)$ where 
\[\Phi(x)=\frac{1}{\sqrt{2\pi}}\int_{-\infty}^x e^{-t^2/2}dt
\] is the cumulative distribution function of the normal distribution. By Hoeffding's inequality, the tail of $P_{r,N}$ is bounded by
\be
\begin{cases}
P_{r,N}\leq \frac{1}{2}\exp\left(-\frac{(2r-N)^2}{2N}\right)&\text{when}\quad r<N/2\\
P_{r,N}\geq 1-\frac{1}{2}\exp\left(-\frac{(2r-N)^2}{2N}\right)&\text{when}\quad r>N/2
\end{cases}.
\ee
Therefore, if $r/N\to \rho$ as $N\to\infty$, then 
\be
\begin{cases}
P_{r,N}\leq\frac{1}{2}\exp\left(-\left(\rho-\frac{1}{2}\right)^2N\right)\to0&\text{when}\quad \rho<1/2\\
P_{r,N}\geq1-\frac{1}{2}\exp\left(-\left(\rho-\frac{1}{2}\right)^2N\right)\to1&\text{when}\quad \rho>1/2\\
\end{cases}.
\ee
which is illustrated in Fig. \ref{prob}.
\begin{figure}
\begin{center}
\includegraphics[width=3in]{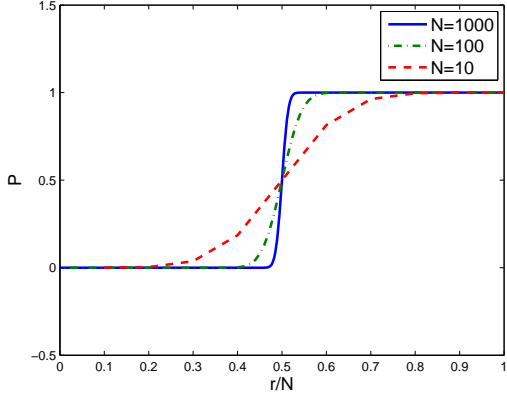}
\end{center}
\caption{$P_{r,N}$ as a function of $r/N$.}\label{prob}
\end{figure}

The discussion above can be summarized by the following theorem, which basically says that if the number of  measurements are more than a half of the size of the signal, the probability that the convex relaxation is exact will tend to $1$ as the size of the signal goes to infinity. 
\begin{theorem}\label{randombinary}
If $u_0$ is a random binary signal with size $N$ without preference on orthants, $A$ is a matrix in general position with rank $r$, when $N$ is large, the probability that $u_0$ can be recovered from linear problem $(P_1)$ can be approximated by $\Phi\left(\frac{2r-N-1}{N-1}\right)$ where $\Phi$ is the cumulative distribution function of the normal distribution. If $\frac{r-1}{N-1}\to\rho>1/2$ as $N\to\infty$, then $u_0$ can be recovered from $(P_1)$ with overwhelming probability at least $1-\frac{1}{2}e^{-c(N-1)}$ where $c=\left(\rho-\frac{1}{2}\right)^2$. 
\end{theorem}

\section{Solving the Optimization Problem}\label{solving}

As stated in section \ref{theorygen}, there are many convex models that can recover the binary signals. Since the measurements might be noisy in practice, we choose to reconstruct the signals via the following optimization problem:
\be\label{P_3}
\min_u \|Au-b\|^2 \quad \st\quad 0\leq u \leq 1.
\ee
First we discuss the robustness of this model. Let $u_0$ is the true binary function and $b=Au_0$ is the clean measurement. Assume $b$ is contaminated by noise $\epsilon$. For corrupted measurement $\tilde b=b+ \epsilon $, we want to investigate if this model will still lead to the correct answer. Let 
\[
B(u)_i=\begin{cases}1,&u_i\geq 1/2\\
0,&u_i<1/2
\end{cases}
\]
be the thresholding operator that maps any function to its closest binary function. The following theorem shows that the model is robust to small perturbation. In section \ref{numerical} the numerical results will show that the more measurements are given, the more robust the reconstruction would be, which is not surprising. 

\begin{theorem}\label{robustness}
If $u_0$ is the unique solution of $(P_1)$, $\tilde b=b+ \epsilon $ is the corrupted measurement,  $\tilde u$ is the minimizer of the optimization problem \[
\min_u \|Au-\tilde b\|_2^2\quad\st\quad 0\leq u\leq 1,\]
 then when $\|\epsilon\|< h(A, \mathbb{O}_{u_0})$ where $h>0$ is a small amount depending only on $A$ and $\mathbb{O}_{u_0}$ (see details in the proof), $B(\tilde u)=u_0$. 
\end{theorem}

Since \eqref{P_3} is a standard  bounded least square problem, it can be solved by many existing optimization algorithms. However, we propose an algorithm that is specifically developed for this problem. It will only utilize the discrete Fourier transform without explicitly storing and multiplying the matrix $A$, which is can be very large in practical problems and can make most out-of-the-box optimization packages very inefficient. 

Our algorithm will be based on the split Bregman method introduced in \cite{Goldstein:2009p2473} and modified in \cite{Szlam:2010p7216} for solving the non-negative least square problem. We replace \eqref{P_3} with an equivalent problem
\[
\min_u \|Au-b\|^2 \quad \st \quad u=P(d).
\]
where $P(d)$ is defined component-wisely by 
\[P(d)=
\begin{cases}1&d\geq 1\\
d&0<d<1\\
0&d\leq 0
\end{cases}.
\]
This constrained problem can be solved iteratively by
\[
\begin{cases}
(d^{k+1},u^{k+1})=\min_{d,u}\frac{\lambda}{2}\|Au-b^k\|_2^2+\|u-P(d)-v^k\|_2^2\\
v^{k+1}=v^k+P(d^{k+1})-u^{k+1}\\
b^{k+1}=b^k+b-Au^{k+1}
\end{cases}
\]
The last two lines are called Bregman steps and can be understood as the gradient ascent steps in the augmented Lagrangian method. Theory on the convergence of this method can be
found in \cite{Kontogiorgis:1998p8034, Osher:2005p632, Goldstein:2009p2473}. The first line can be solved exactly respectively on $d$ and $u$, giving rise to the following iterations:
\be
\begin{cases}
d^{k+1}=P(u^k-v^k)\\
u^{k+1}=(\lambda A^\top A+I)^{-1}(\lambda A^\top b^k+P(d^{k+1})+v^k)\\
v^{k+1}=v^k+P(d^{k+1})-u^{k+1}\\
b^{k+1}=b^k+b-Au^{k+1}
\end{cases}
\ee
Here the first, third and last lines contain only trivial computations. For the second line, we can notice that when $A=S \F$, $(\lambda A^\top A+I)^{-1}=\F^{-1} (N\lambda S+I)^{-1}\F$ where $N$ is the size of the signal. Since $N\lambda S+I$ is nothing but a diagonal matrix, the whole operator can be calculated efficiently and precisely. Therefore we get the following algorithm \ref{OverallAlgorithm}. Numerical results in the next section would show that this algorithm works very well.

If the given information is not the partial Fourier measurements of the signal but a filtered signal, i.e. $b=Au=\F^\top K\F u$ where $K$ is a filter in frequency domain, then this algorithm still works after a small modification. The only point that needs to be changed is that now $(\lambda A^\top A+I)^{-1}=\F^{-1} (N\lambda K^\top K+I)^{-1}\F$. 

In some problems, the norm $\|Au-b\|^2$ can also be preconditioned, e.g.  to prevent the effect of the noise in high frequencies. We can minimize $\|Au-b\|_M^2=(Au-b)^\top M (Au-b)$ where $M$ is a certain preconditioner in the frequency domain. Again, the algorithm still works without many modifications except the inverse operator becoming $(\lambda A^\top MA+I)^{-1}$ now.

\begin{algorithm}
\caption{The Split Bregman Algorithm for Solving \eqref{P_3}}\label{OverallAlgorithm}
\begin{algorithmic}

\STATE  Initialize: Let $b_0=b$. Start from initial guess $u=\F^{-1}b$.

\WHILE {$\|Au-b_0\|^2$ not small enough}

\STATE $d\leftarrow P(u-v)$

\STATE $u\leftarrow (\lambda A^\top A+I)^{-1}(\lambda A^\top b+P(d)+v)$

\STATE $v\leftarrow v+P(d)-u$

\STATE $b\leftarrow b+b_0-Au$

\ENDWHILE

\STATE $u \leftarrow B(u)$. 

\end{algorithmic}
\end{algorithm}

\section{Numerical Results}\label{numerical}

First we numerically show that binary signals can in general be characterized by very few frequency measurements. Fig. \ref{Fig1} shows several 1D and 2D signals:
\begin{itemize}
\item  The first one is a 1D binary signal. It contains 15 constant intervals with value 1 and 15 intervals with 0, so by theorem \ref{lowfreq1D} it can be fully determined by Fourier coefficients $a_k$ for $|k|\leq 16$, no matter how long the signal actually is (the length of the signal shown here is 400). 
\item  The second one is a binary image corresponding to a geometrical shape, the size is $200\times 200$.  The experiment shows that it is fully determined by Fourier coefficients $a_k$ for $|k|\leq 5$. The needed Fourier coefficients for characterization of the shape account for 0.2\% of the total Fourier coefficients.
\item  The third one is a barcode image, which can also be treated essentially as a 1D signal. It contains 15 black bars and 15 white bars, so the Fourier coefficients needed are $a_k$ with $|k_1|\leq 16$, where $k_1$ is the component of $k$ along the horizontal dimension, no matter how large the image actually is  (the width of the barcode shown here is 400).
\item  The last one is an image with handwritten letters,  with size $100\times 100$.  The experiment shows that it is fully determined by Fourier coefficients $a_k$ with $|k|\leq 10$. The needed Fourier coefficients for the characterization of the image account for 3.17\% of the total Fourier coefficients.

\end{itemize}

\begin{figure}[]
\begin{center}
\includegraphics[width=0.8in]{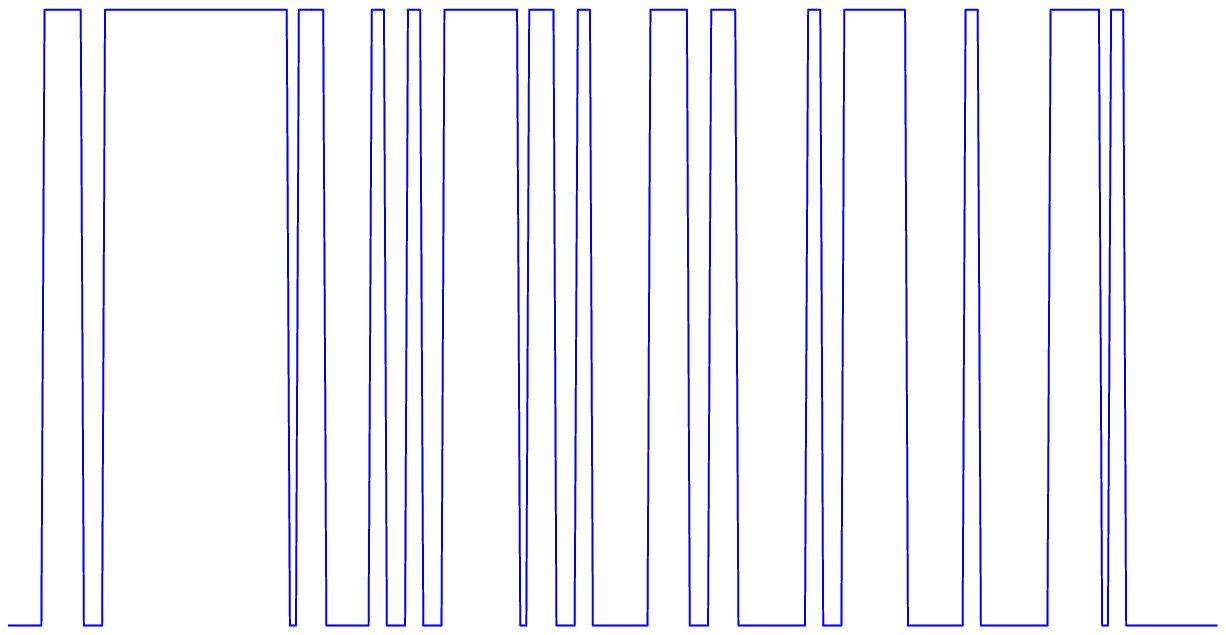}
\includegraphics[width=0.8in]{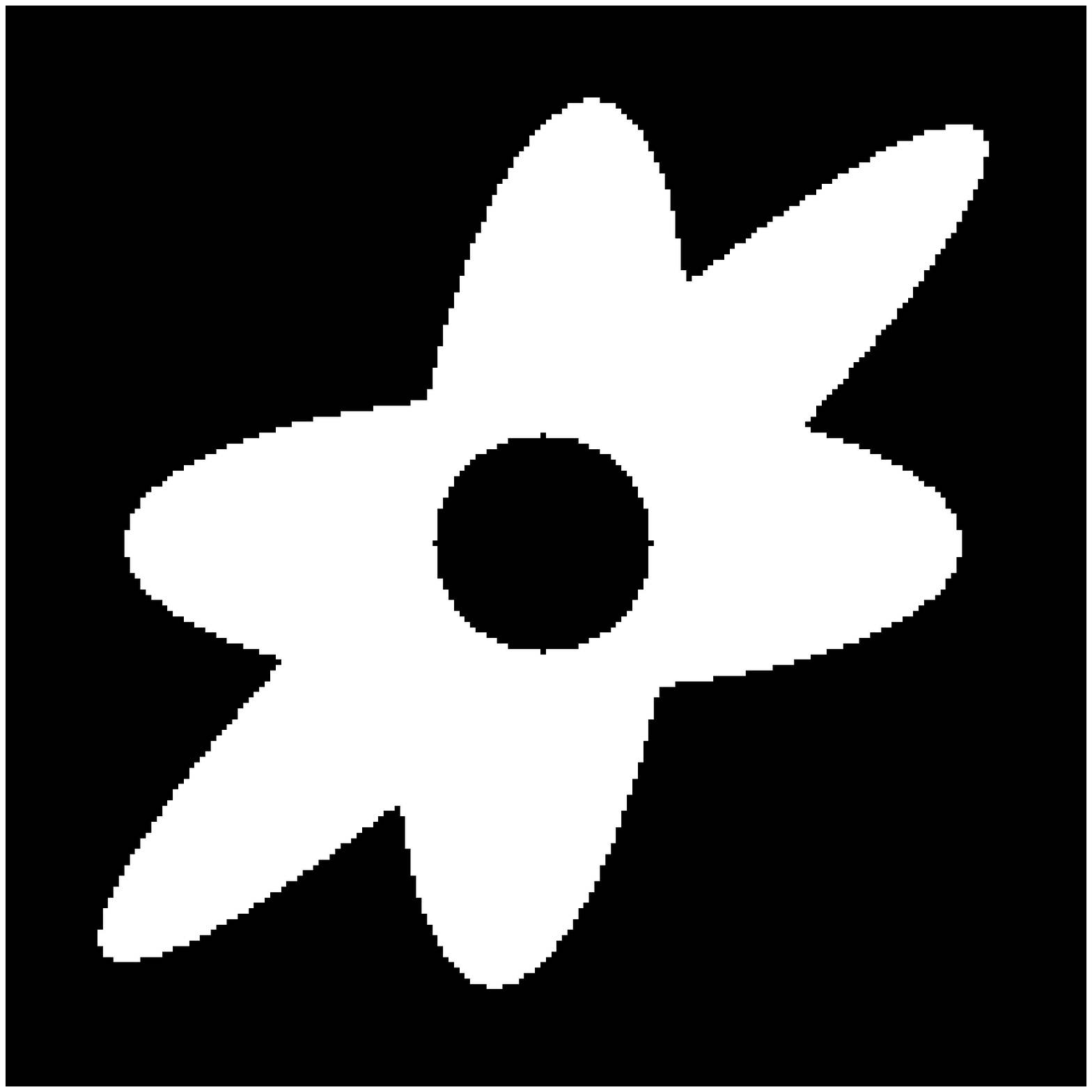}
\includegraphics[width=0.8in]{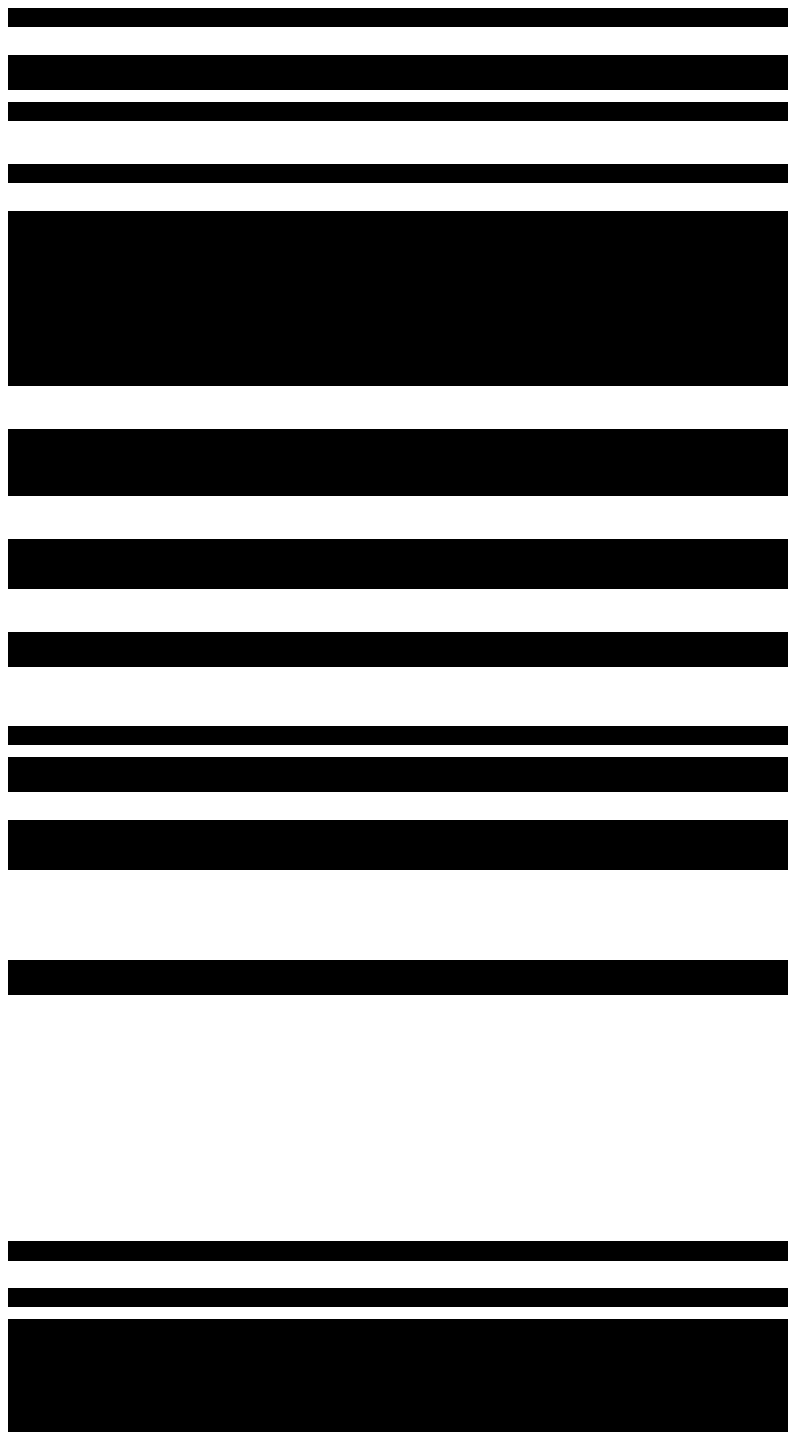}
\includegraphics[width=0.8in]{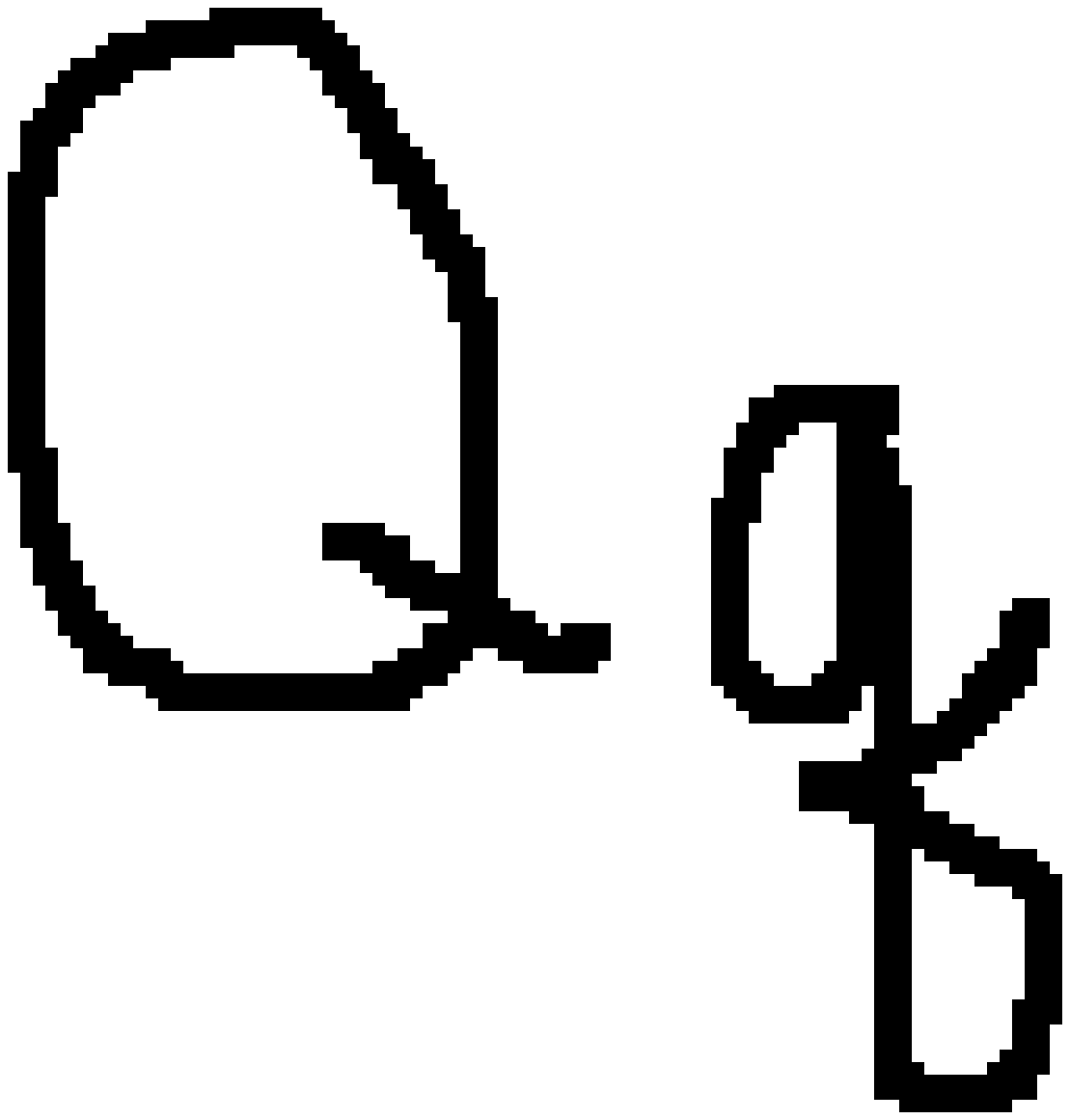}
\end{center}
\caption{Several binary signals: a random 1D binary signal, a geometrical shape, a barcode and a handwriting image.}\label{Fig1}
\end{figure}

To show that these binary signal can be recovered from the low frequency measurements, we filter them with a low-pass Gaussian kernel whose band corresponds to the partial frequencies. As long as the needed low frequency information is precisely given, an exact reconstruction would be available. However, the numerical experiments also show that the more measurements are known, the faster the reconstruction is, which means that it is easier for the algorithm to find the correct binary signal. Fig. \ref{Fig2} demonstrates the filtered signal and the reconstruction. The curves in Fig. \ref{time} show that the reconstruction time decreases when more measurements are given. The $x$-axis measures the radius of the support of the given frequency information, while the $y-$axis measures the logarithm of computational time.

\begin{figure}[]
\begin{center}
\includegraphics[width=0.8in]{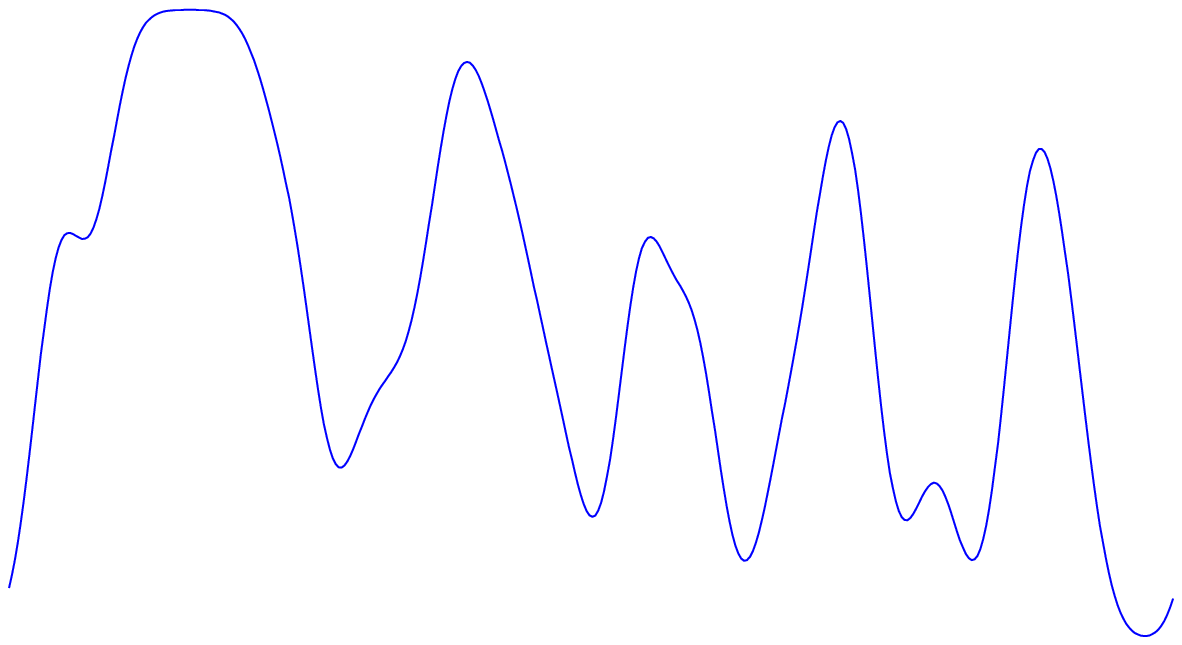}
\includegraphics[width=0.8in]{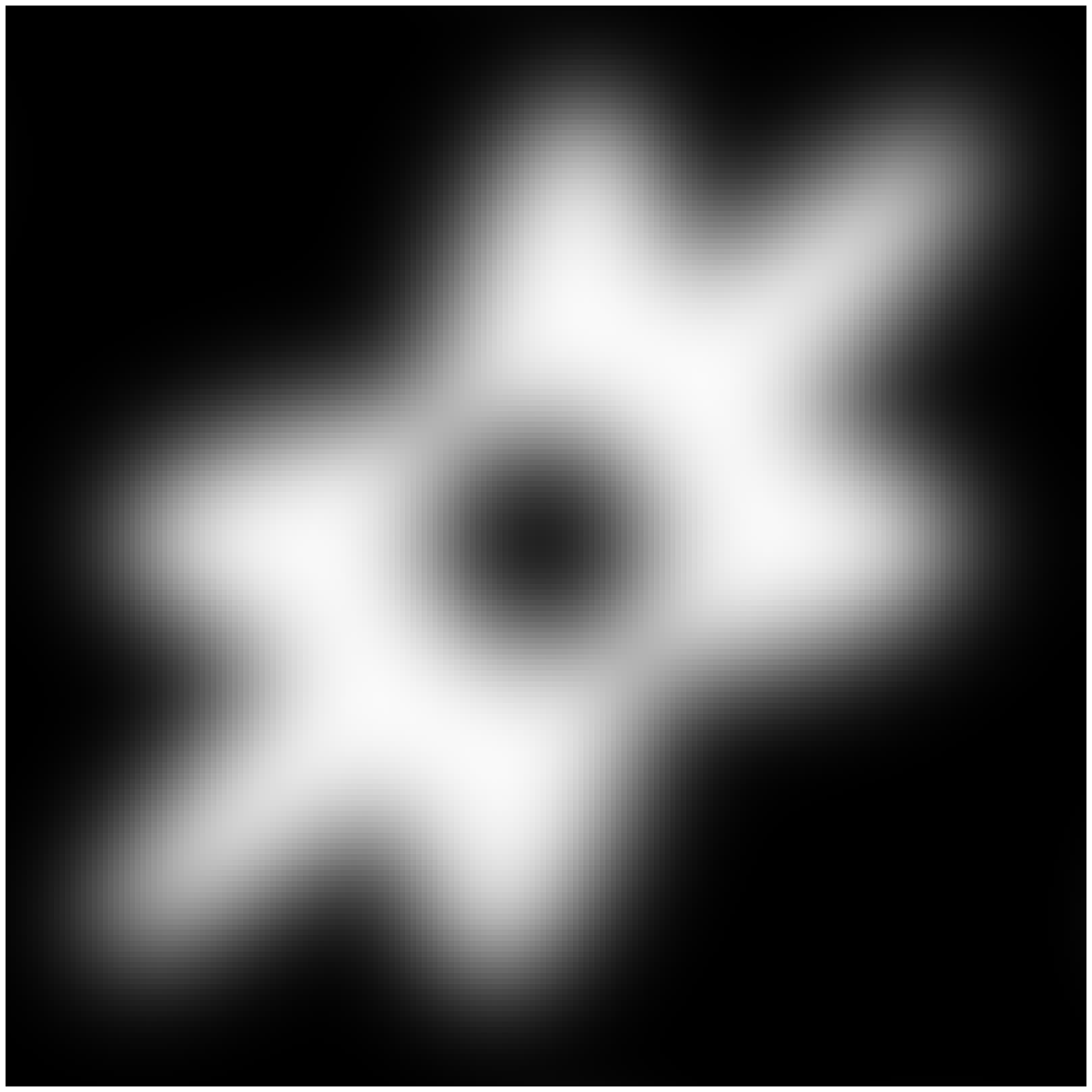}
\includegraphics[width=0.8in]{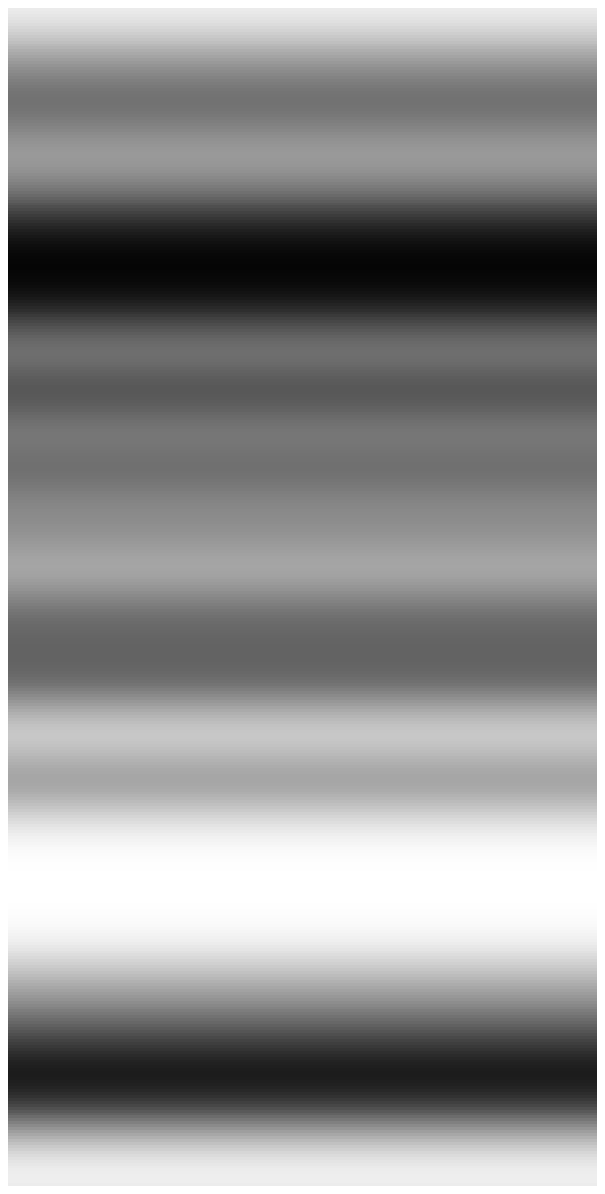}
\includegraphics[width=0.8in]{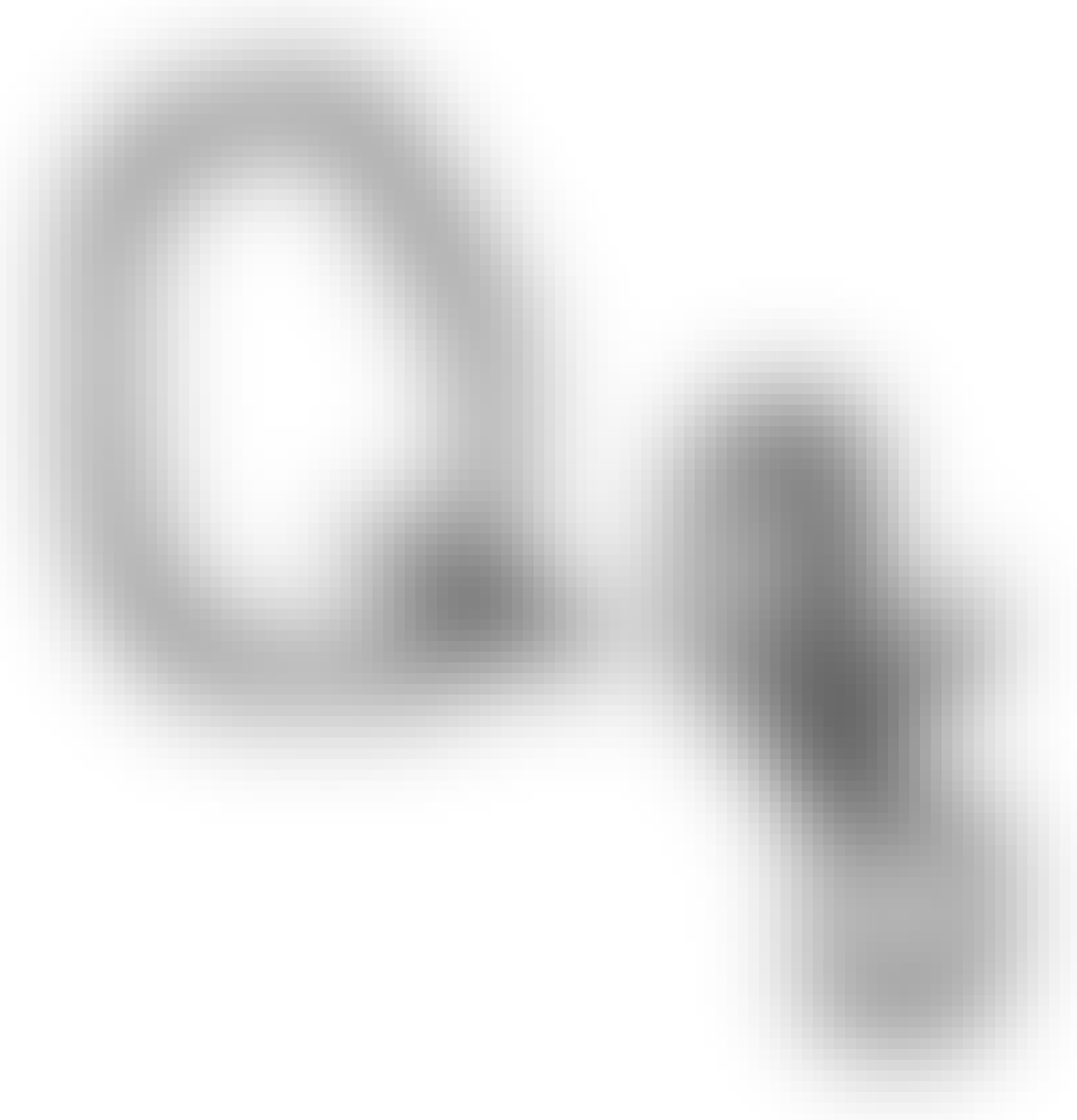}\\
\includegraphics[width=0.8in]{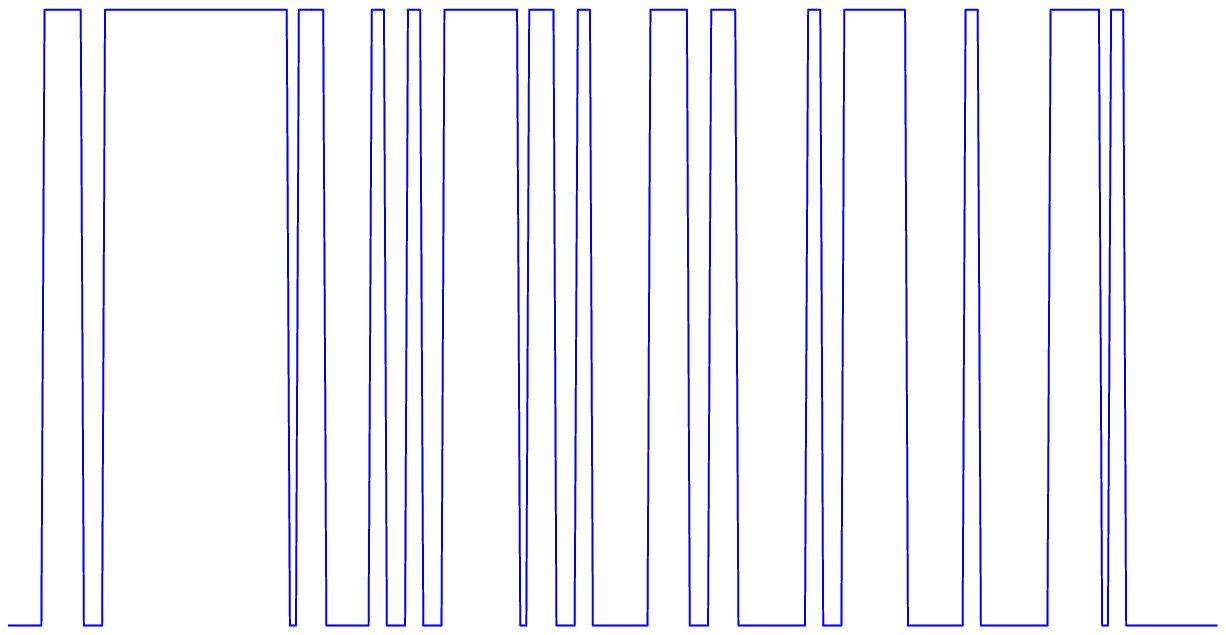}
\includegraphics[width=0.8in]{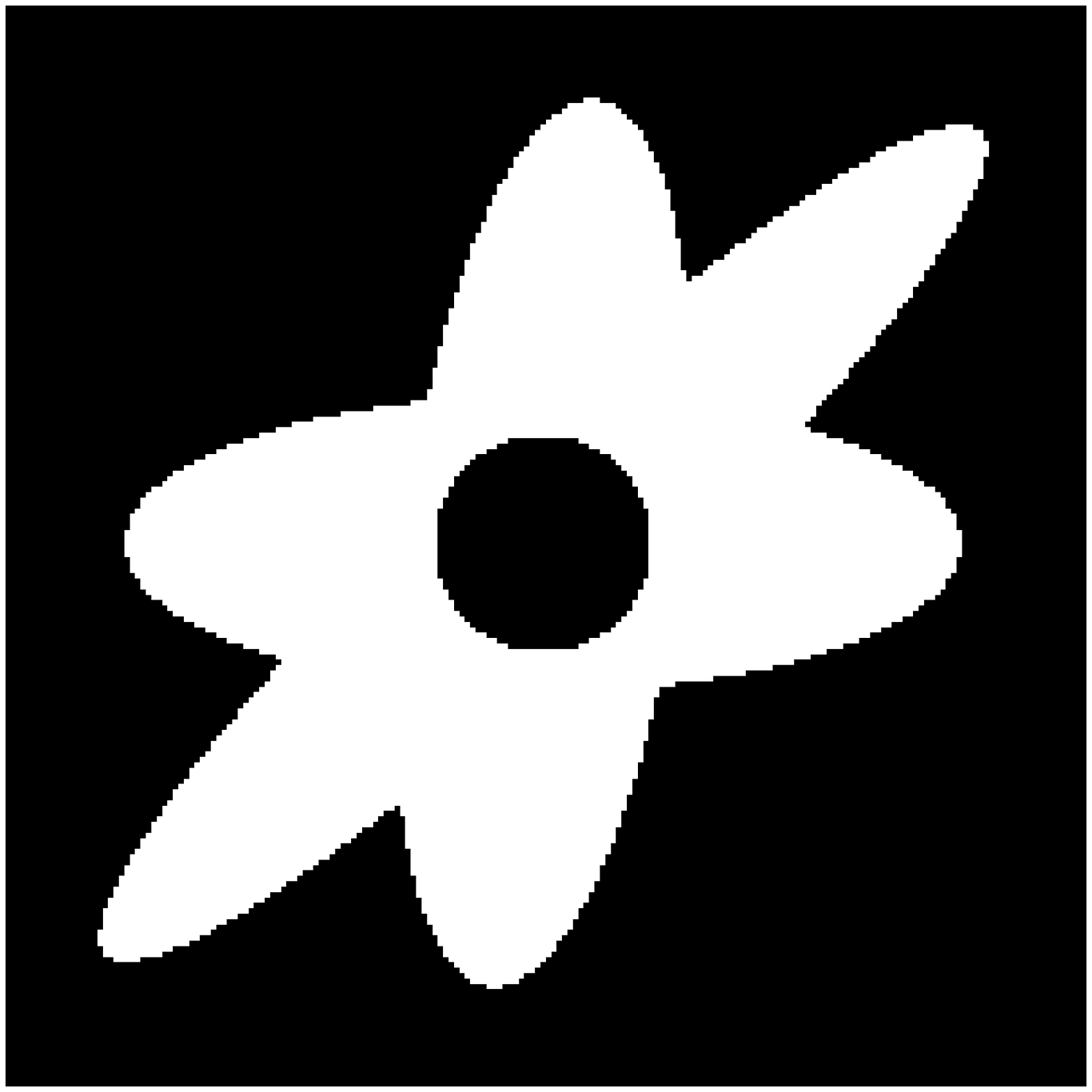}
\includegraphics[width=0.8in]{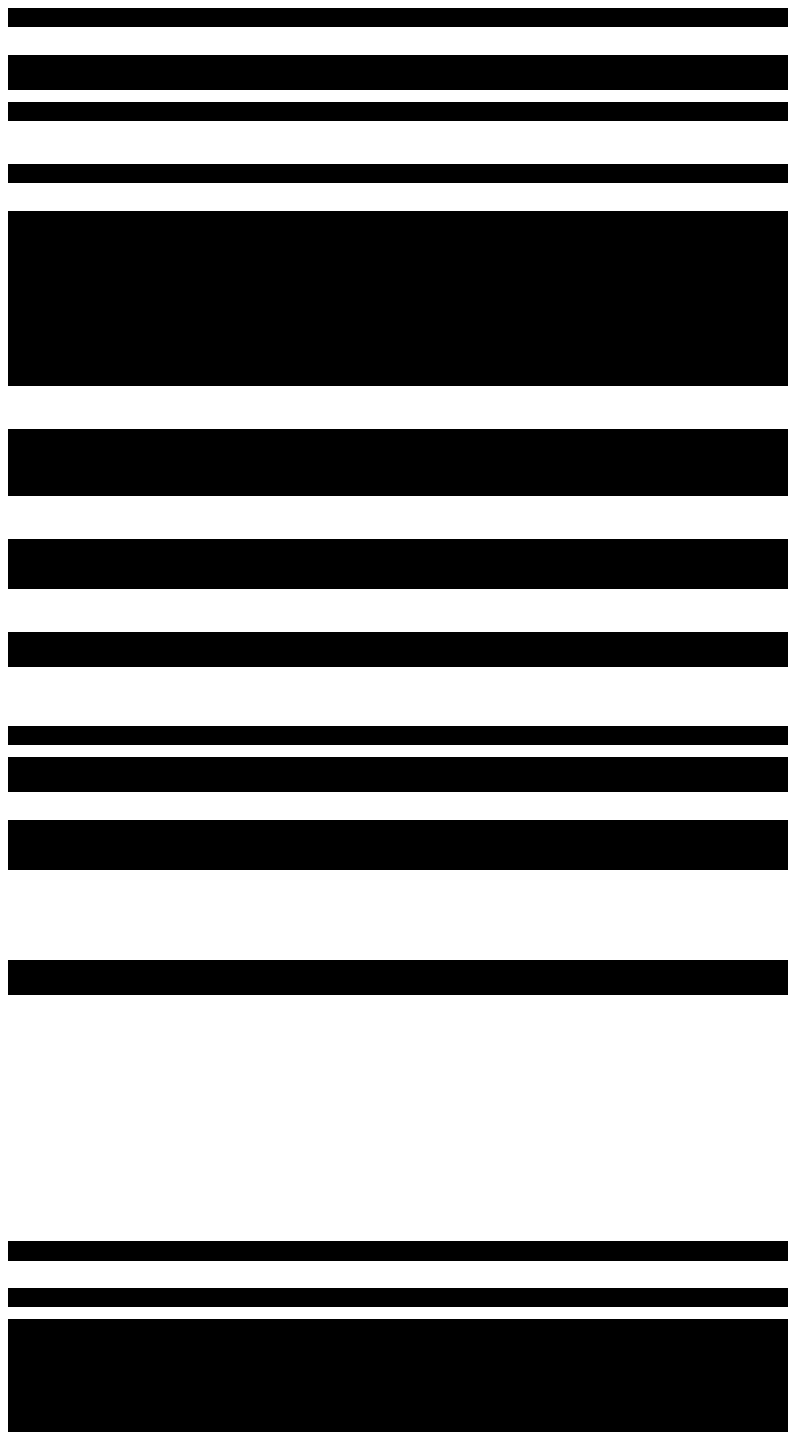}
\includegraphics[width=0.8in]{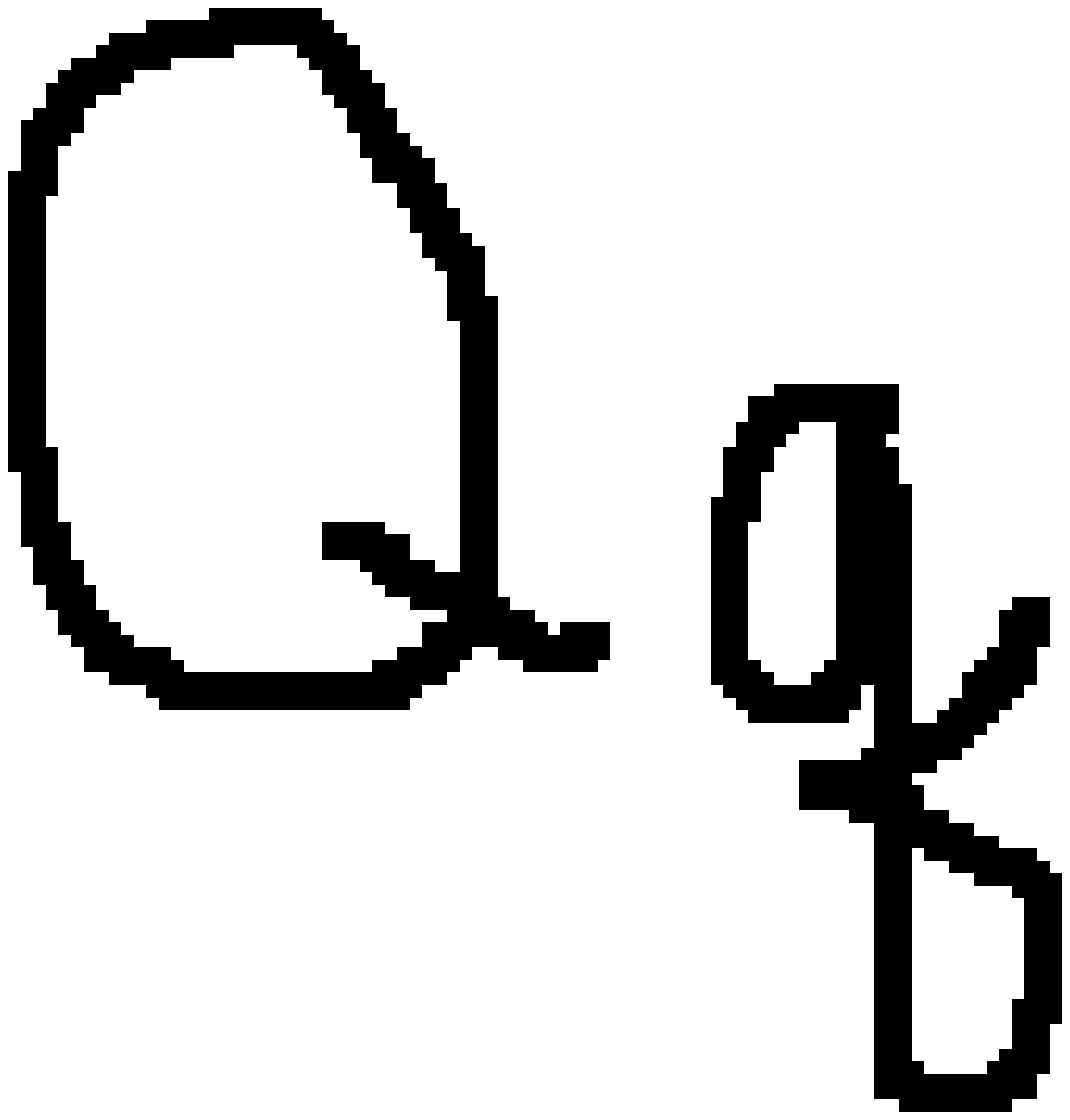}
\end{center}
\caption{The signals after low-pass filtering and the reconstructions.}\label{Fig2}
\end{figure}

\begin{figure}[]
\begin{center}
\includegraphics[width=3in]{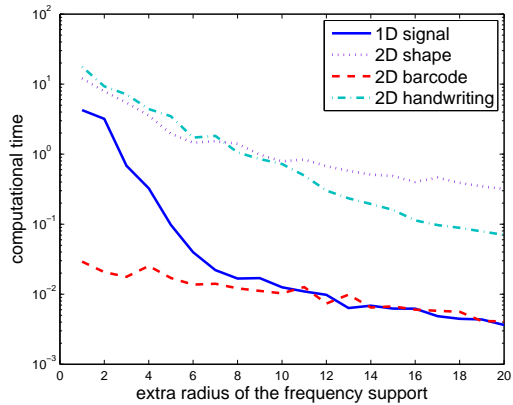}
\end{center}
\caption{Time elapsed (in seconds) for the reconstruction versus the extra radius of the support of the given frequency information.}\label{time}
\end{figure}

To demonstrate that the reconstruction is robust, we now recover the signal from the low-pass filtered measurements with noise. It is well known that deblurring with large amounts of noise present is a difficult task. Our results show that even with very large amounts of noise and strong blurring, the results are still sensible.  In Fig. \ref{Fig4}, tests on a 1D binary signal with noisy measurements are demonstrated. When the input signal is very noisy, the positions of the bars in the reconstruction are not precisely equal but close to the original signal. The minor difference between the reconstruction and the true signal are highlighted by circles. Fig. \ref{Fig5} demonstrates how the number of miss-identified 0s and 1s changes with the number of measurements and noise level. We do the experiments for different levels of noise and different numbers of measurements respectively. For any given pair of fixed noise level and number of measurements, 1000 random tests are taken to get the average number of miss-identified 0s and 1s. In Fig. \ref{Fig6}-\ref{Fig8}, each group shows a signal with a different level of noise. The parameters of the blurring kernel and the noise levels are given in the captions. The computational costs are also recorded. All computations are done in Matlab on a 2.8 GHz Intel CPU.

\begin{figure}[]
\begin{center}
\includegraphics[width=3.5in, height=2.5in]{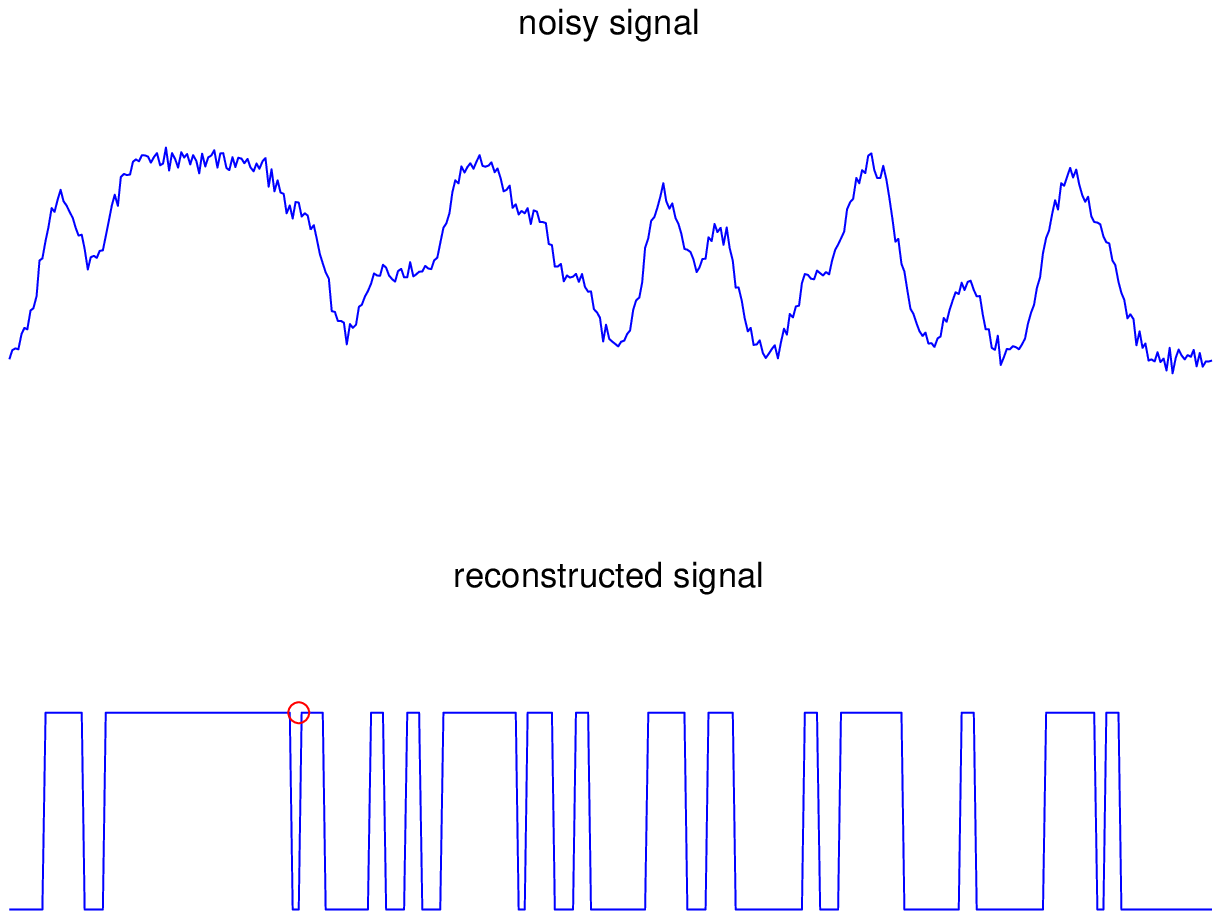}
\includegraphics[width=3.5in, height=2.5in]{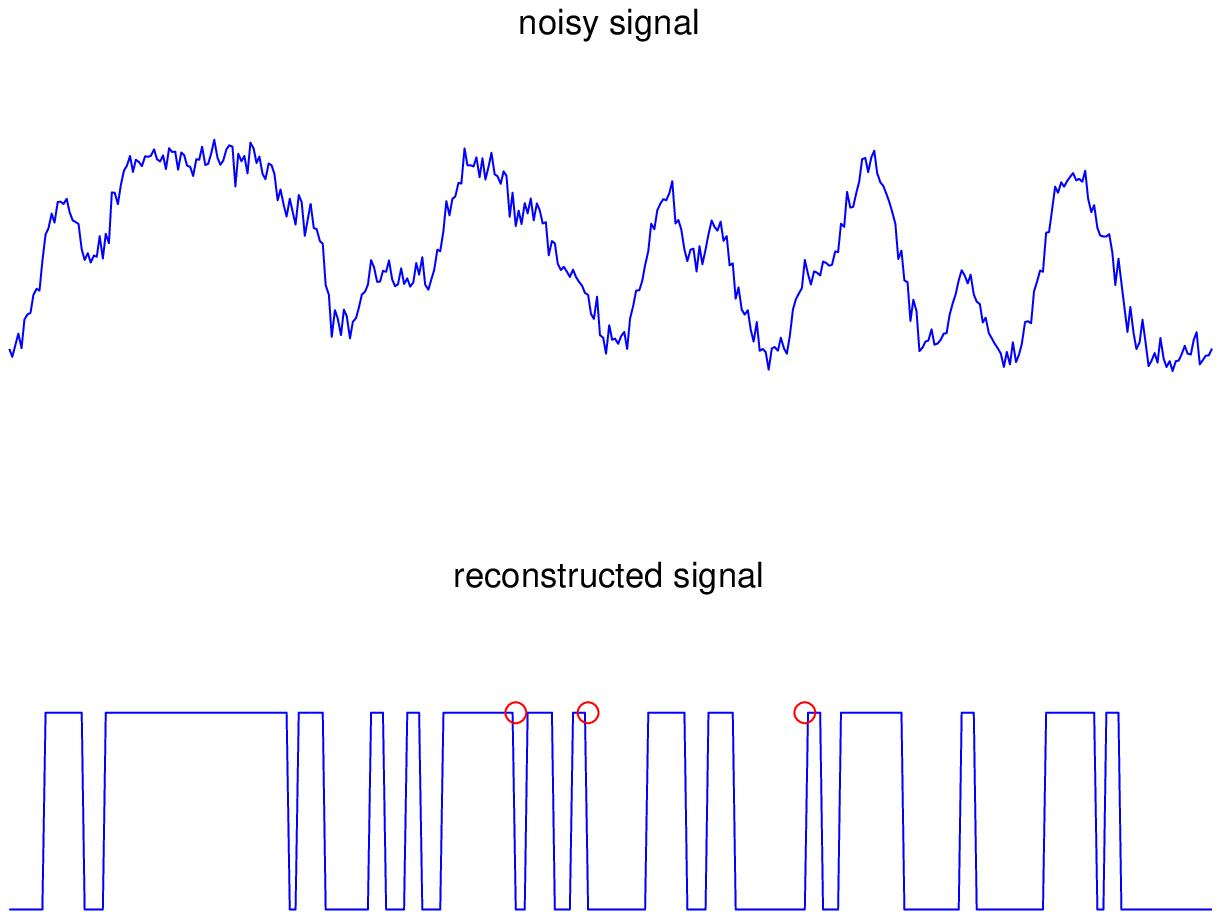}
\includegraphics[width=3.5in, height=2.5in]{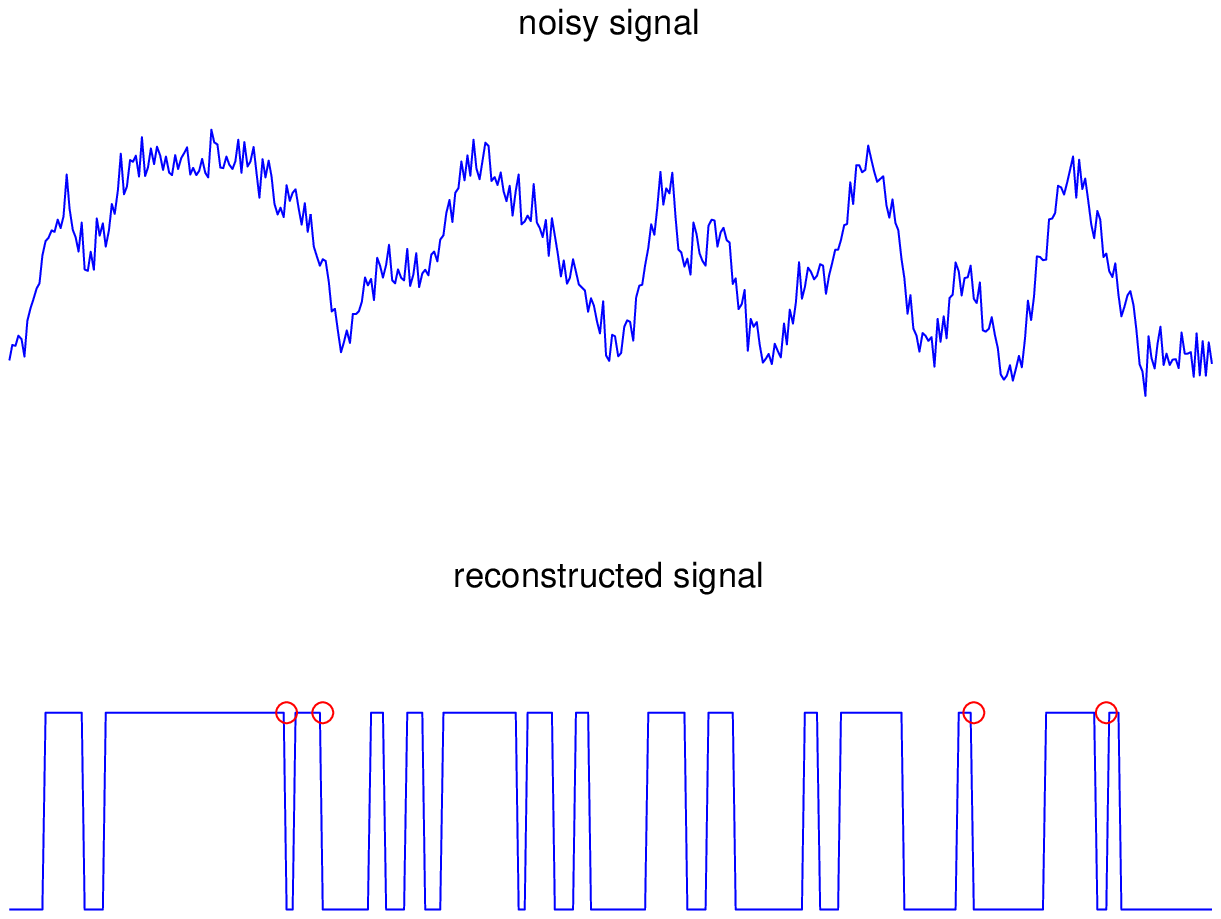}
\end{center}
\caption{The 1D signal is blurred with a Gaussian filter with $\sigma=5$ generated by Matlab command \texttt{fspecial}. The standard deviation of Gaussian noise added is respectively $0.03$, $0.05$, $0.07$, generated by Matlab command \texttt{randn}. The positions of the bars in the reconstruction are sometime not precisely equal to the original due to the present of the noise. The minor difference between the reconstruction and the true signal are highlighted by circles. The average computational time for reconstructions are respectively 0.05s, 0.04s, 0.03s.}\label{Fig4}
\end{figure}

\begin{figure}[]
\begin{center}
\includegraphics[width=3in, height=2in]{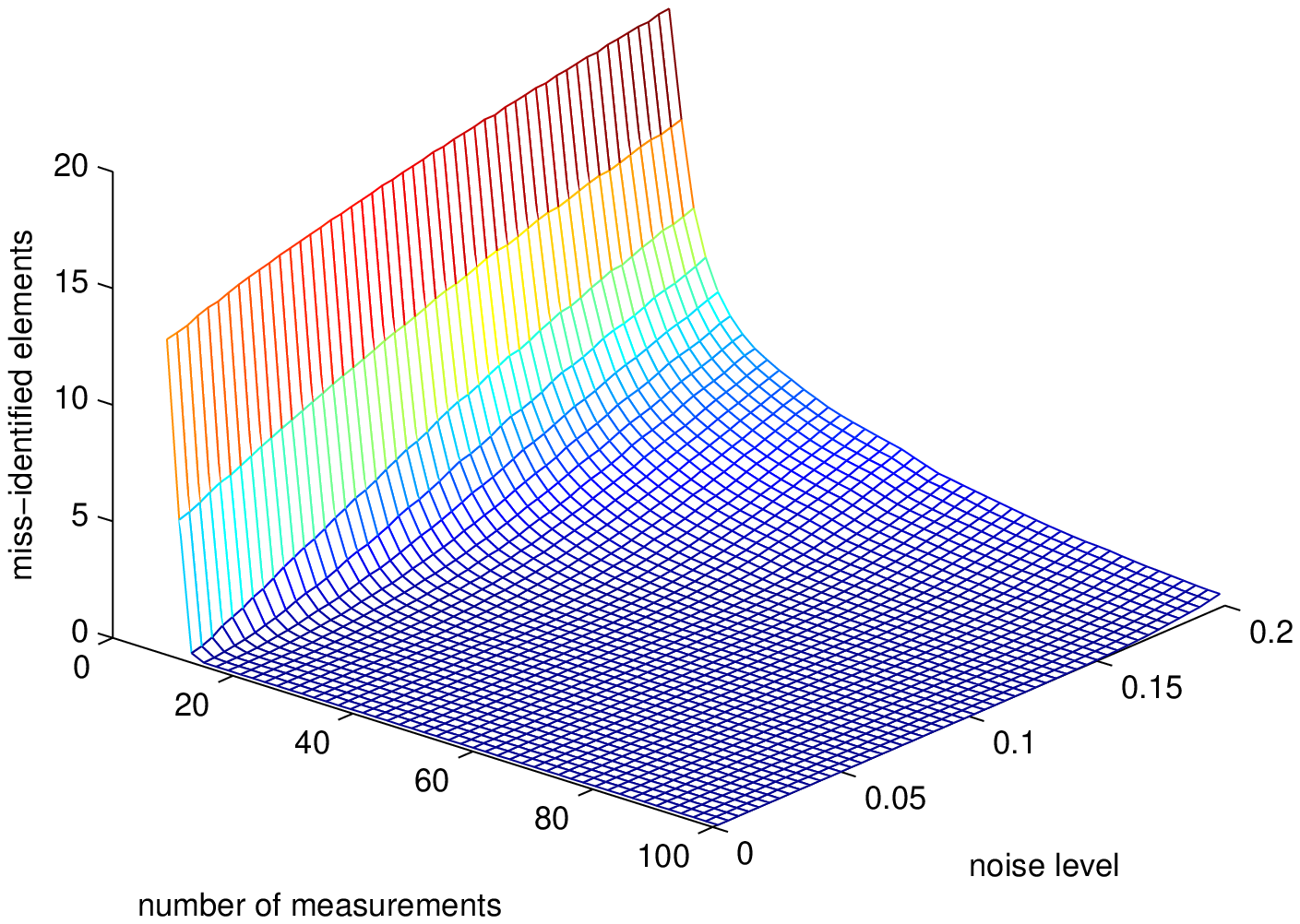}
\includegraphics[width=3in, height=2in]{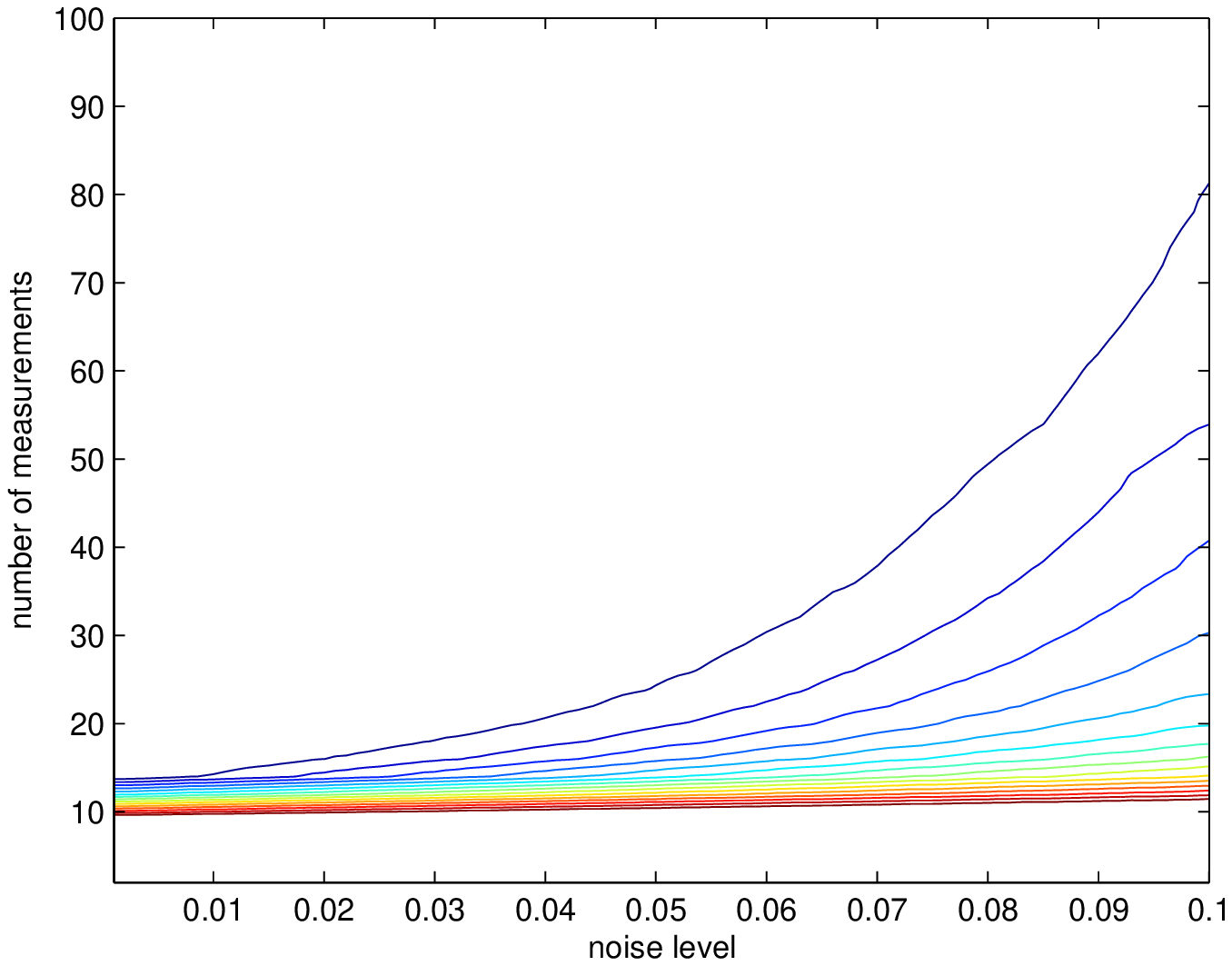}
\end{center}
\caption{The trade-off between number of measurements vs. the noise level obtained from empirical tests. The noise level is defined as the standard deviation of the gaussian noise. For each pair of fixed noise level and number of measurements, 1000 random tests are taken to get the average number of miss-identified 0s and 1s. In each test the signal length is 100 and the number of consecutive constant intervals is 10. The right figure shows the levelset curve of the left figure. Different curves correspond to different number of miss-identified 0s and 1s.}\label{Fig5}
\end{figure}

\begin{figure}[]
\begin{center}
\includegraphics[width=1in]{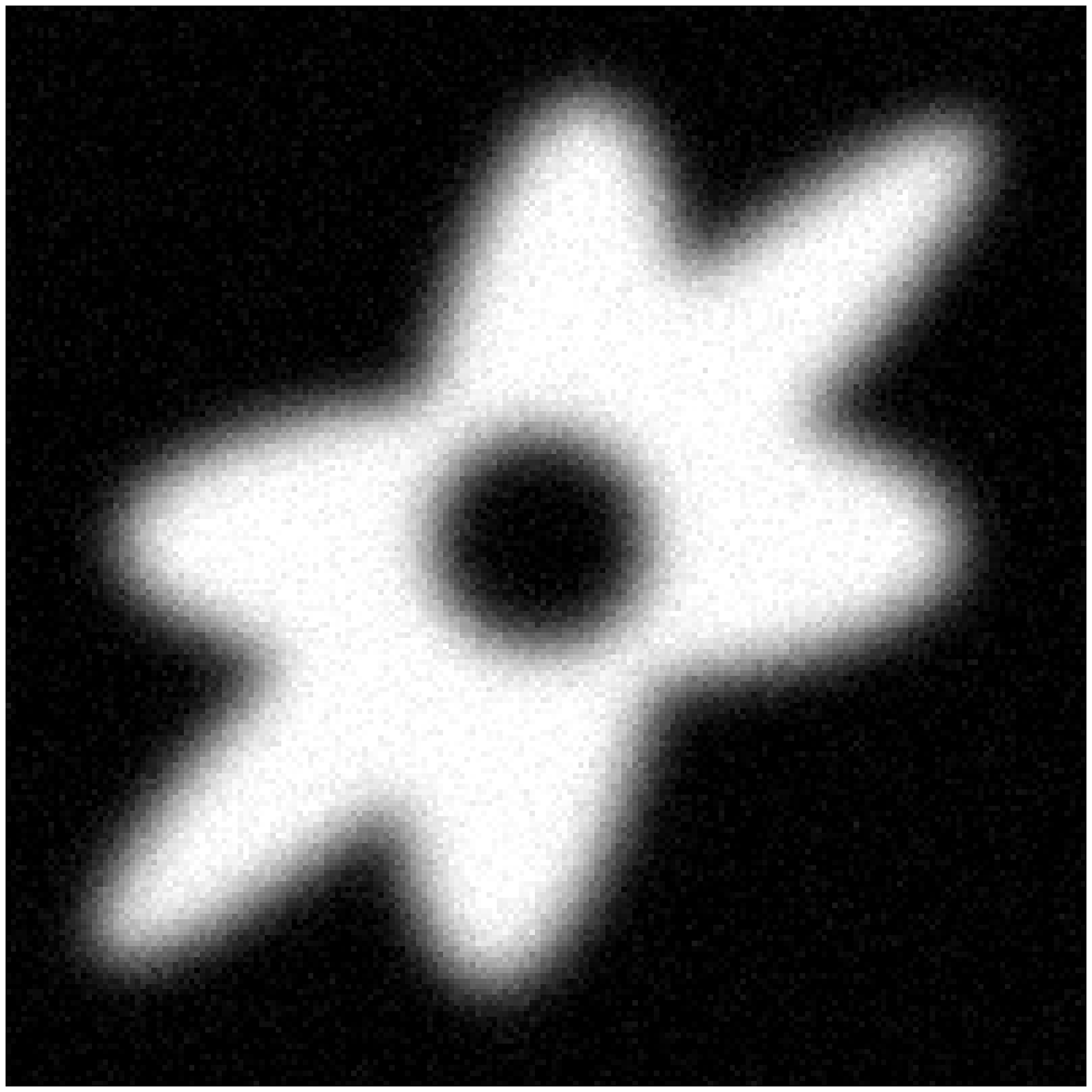}
\includegraphics[width=1in]{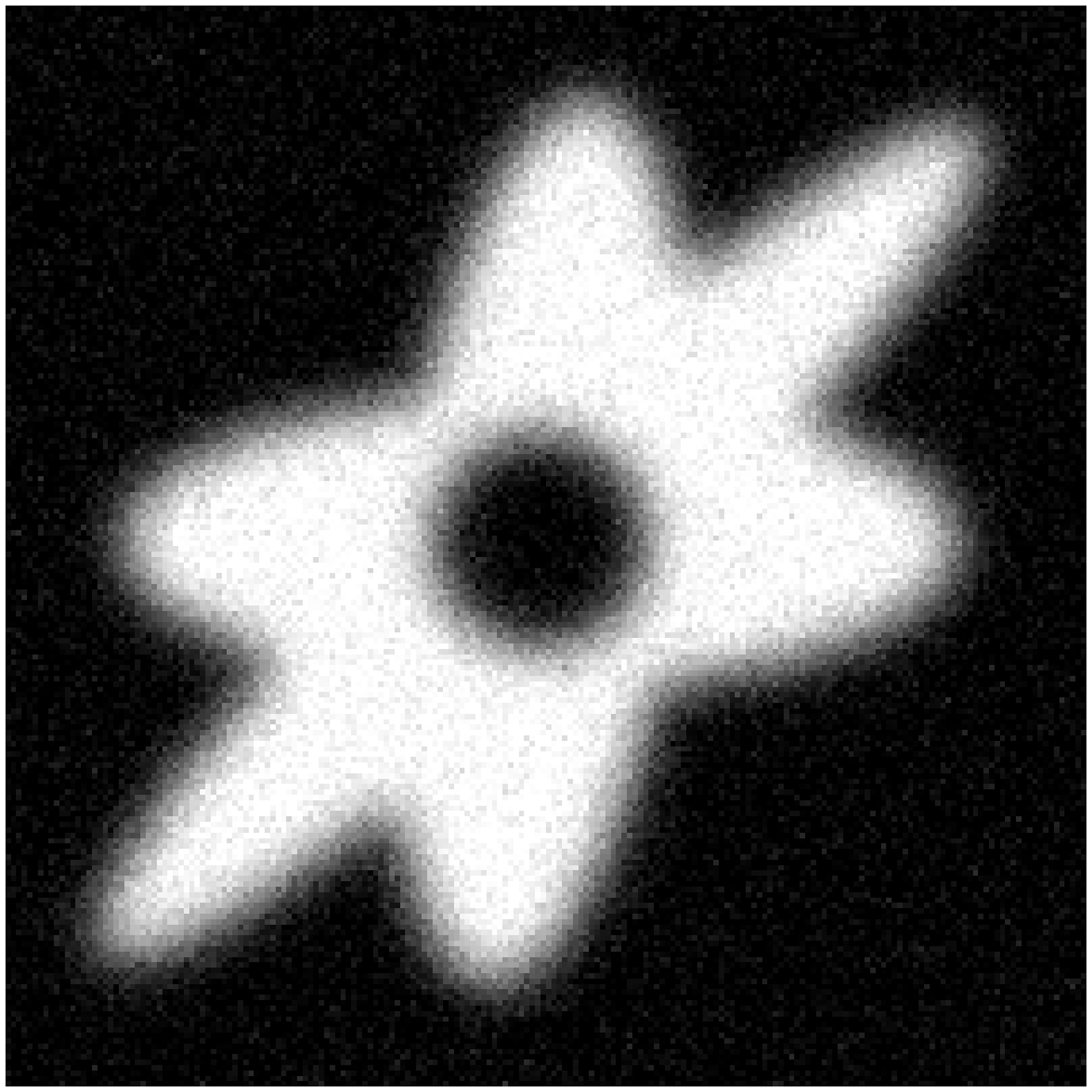}
\includegraphics[width=1in]{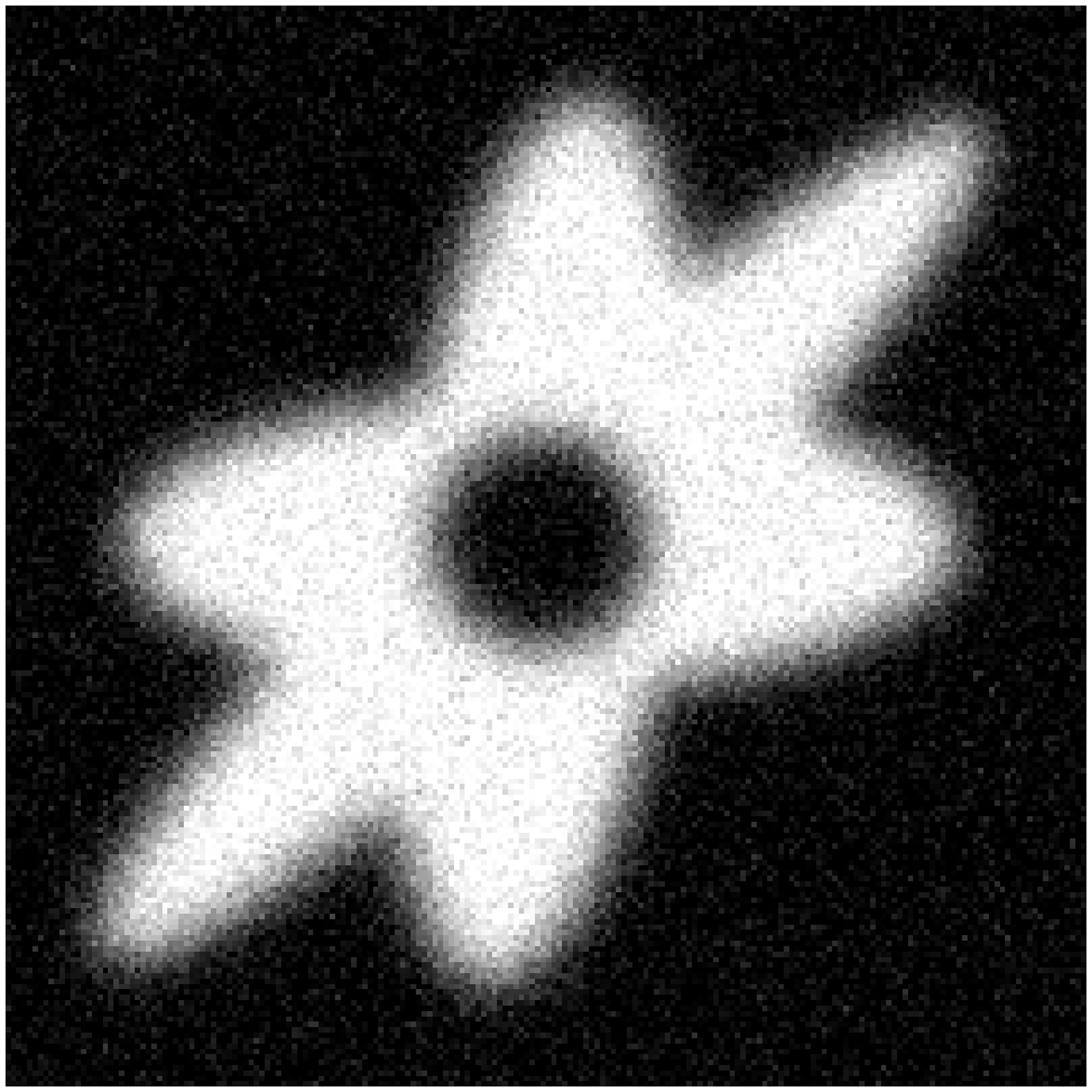}\\
\includegraphics[width=1in]{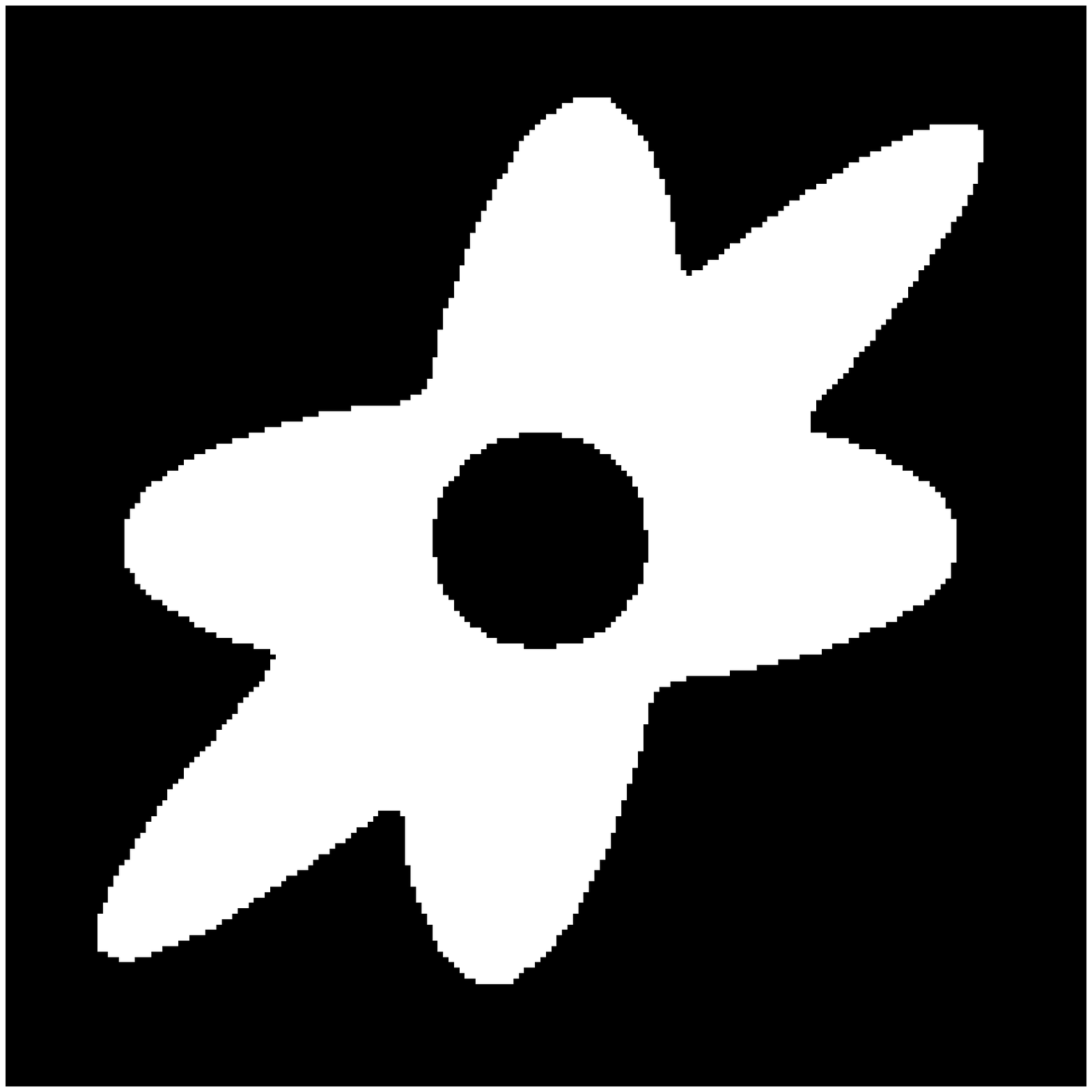}
\includegraphics[width=1in]{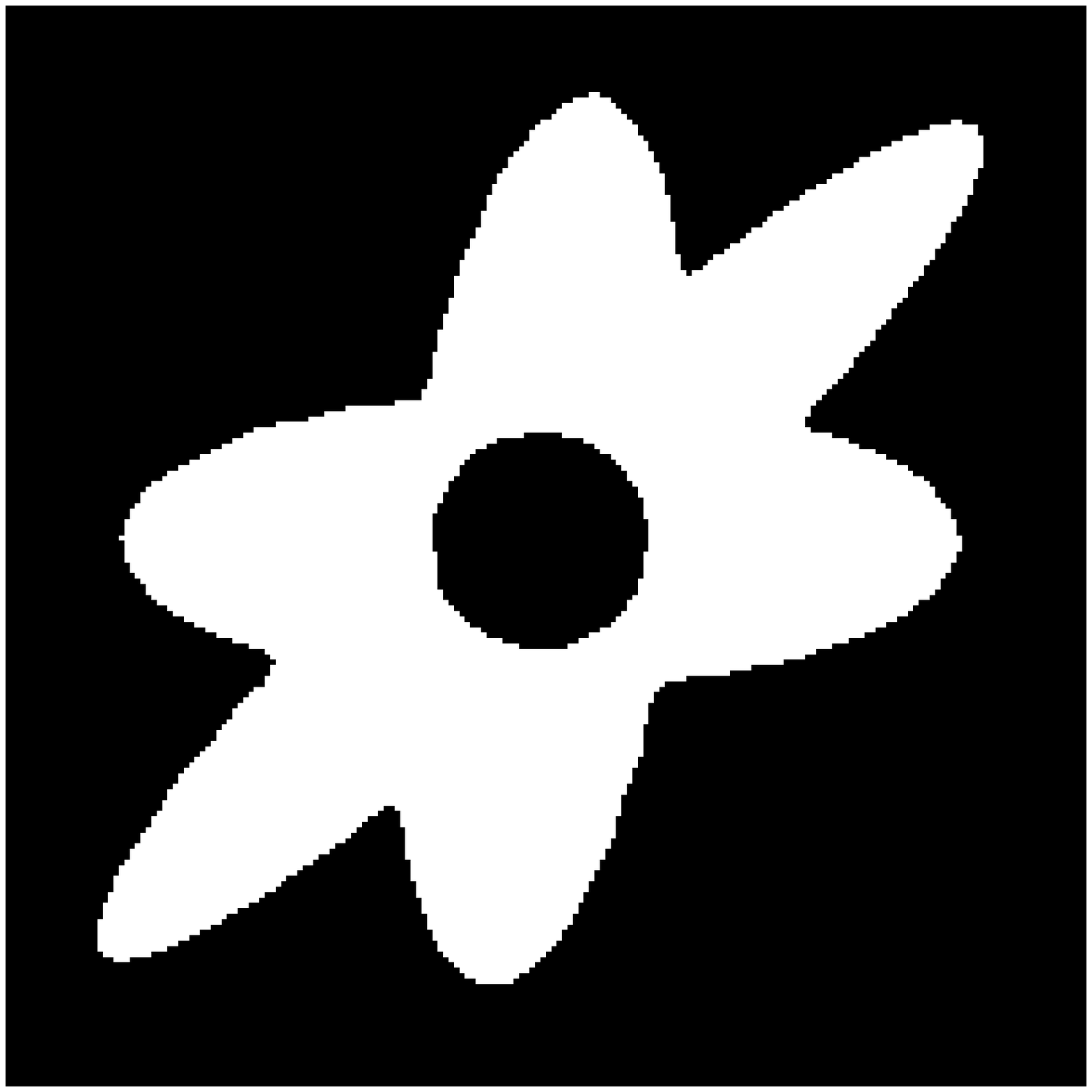}
\includegraphics[width=1in]{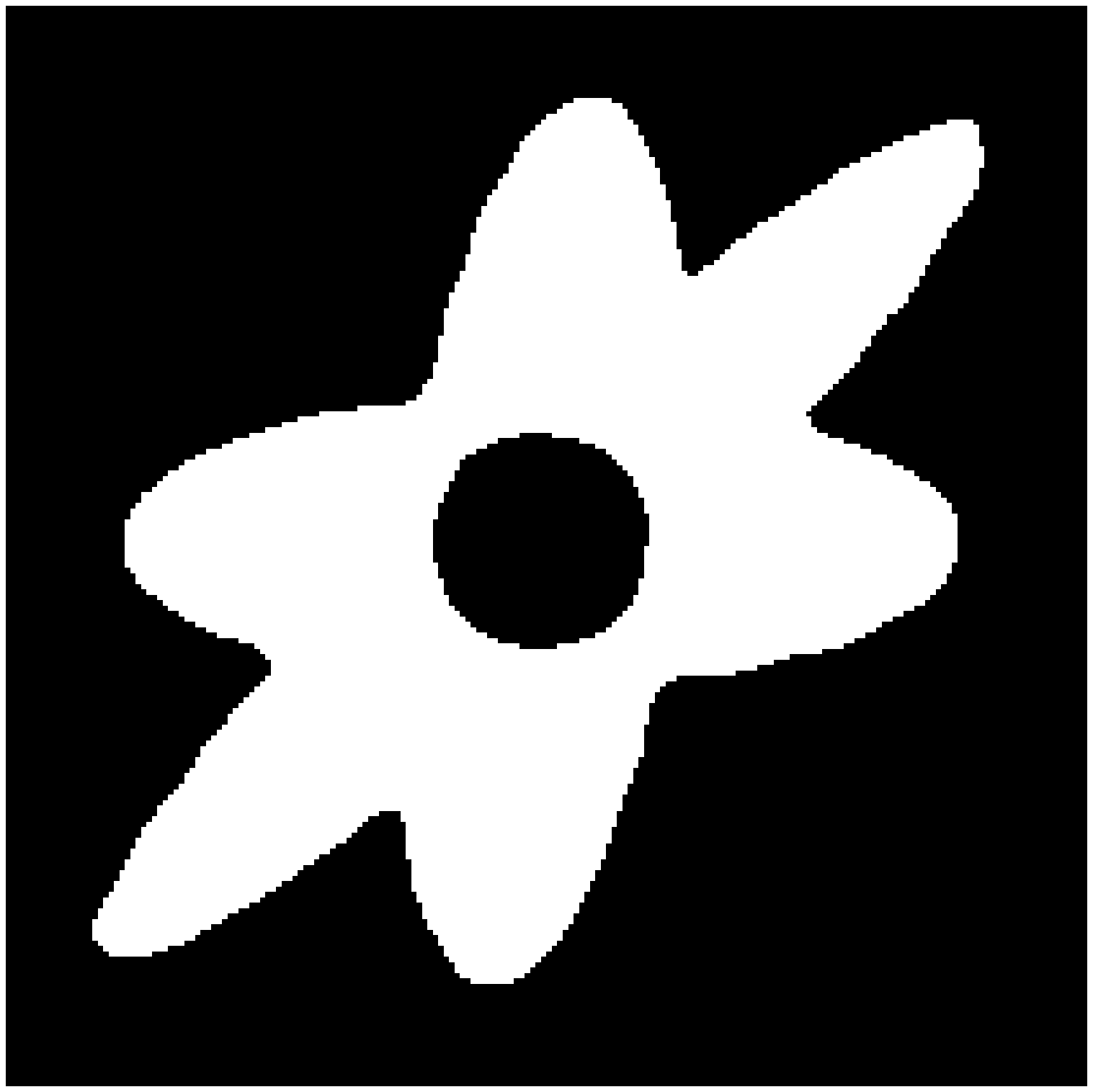}
\end{center}
\caption{The first image blurred with a Gaussian filter with $\sigma=5$,  generated by the Matlab command \texttt{fspecial}. The amplitude of Gaussian noise is respectively $0.03$, $0.05$, $0.07$, generated by the Matlab command \texttt{randn}. The average computational time for the reconstructions is respectively 0.81s, 0.81s, 0.80s.}\label{Fig6}
\end{figure}

\begin{figure}[]
\begin{center}
\includegraphics[width=1in]{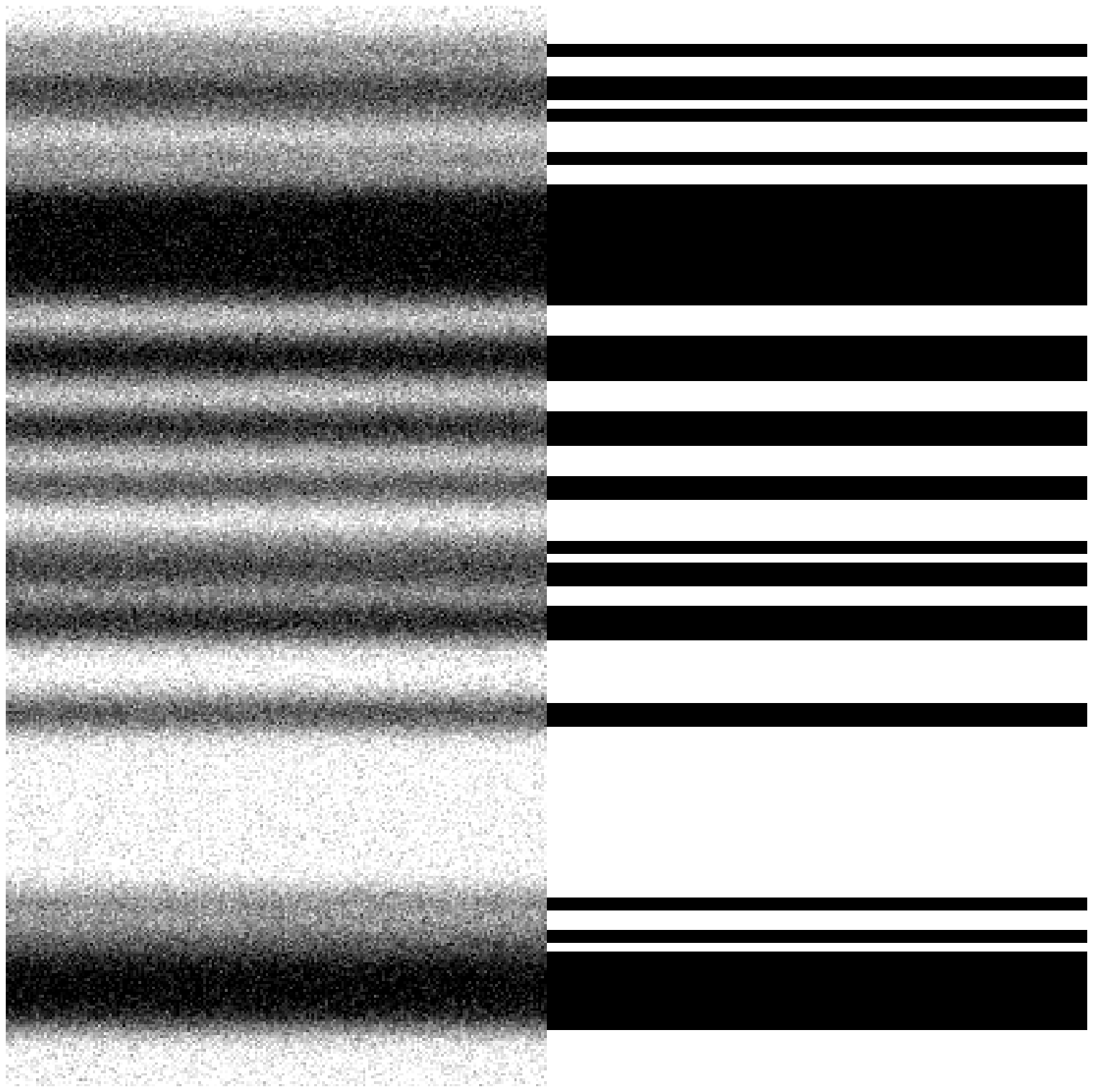}
\includegraphics[width=1in]{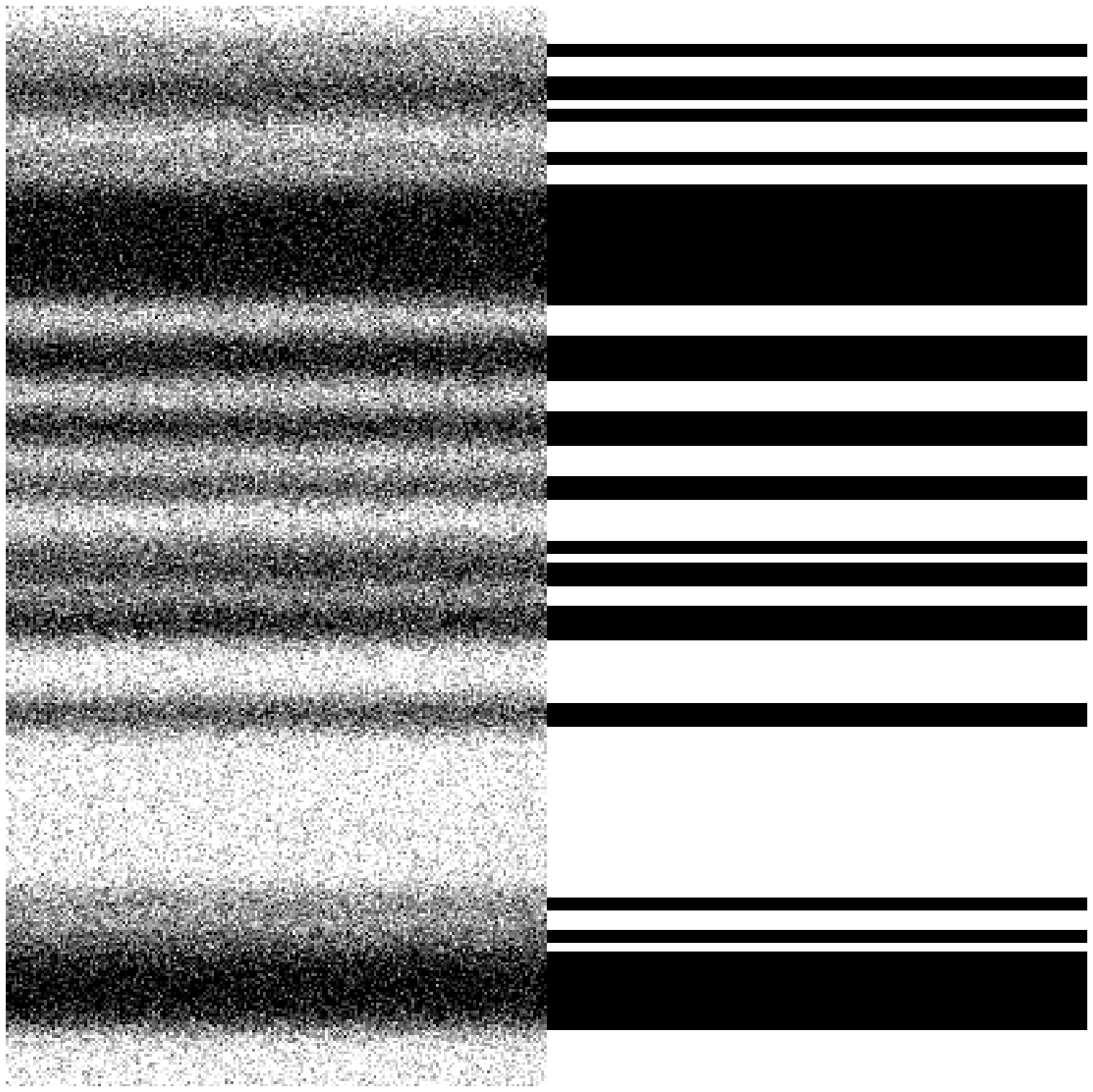}
\includegraphics[width=1in]{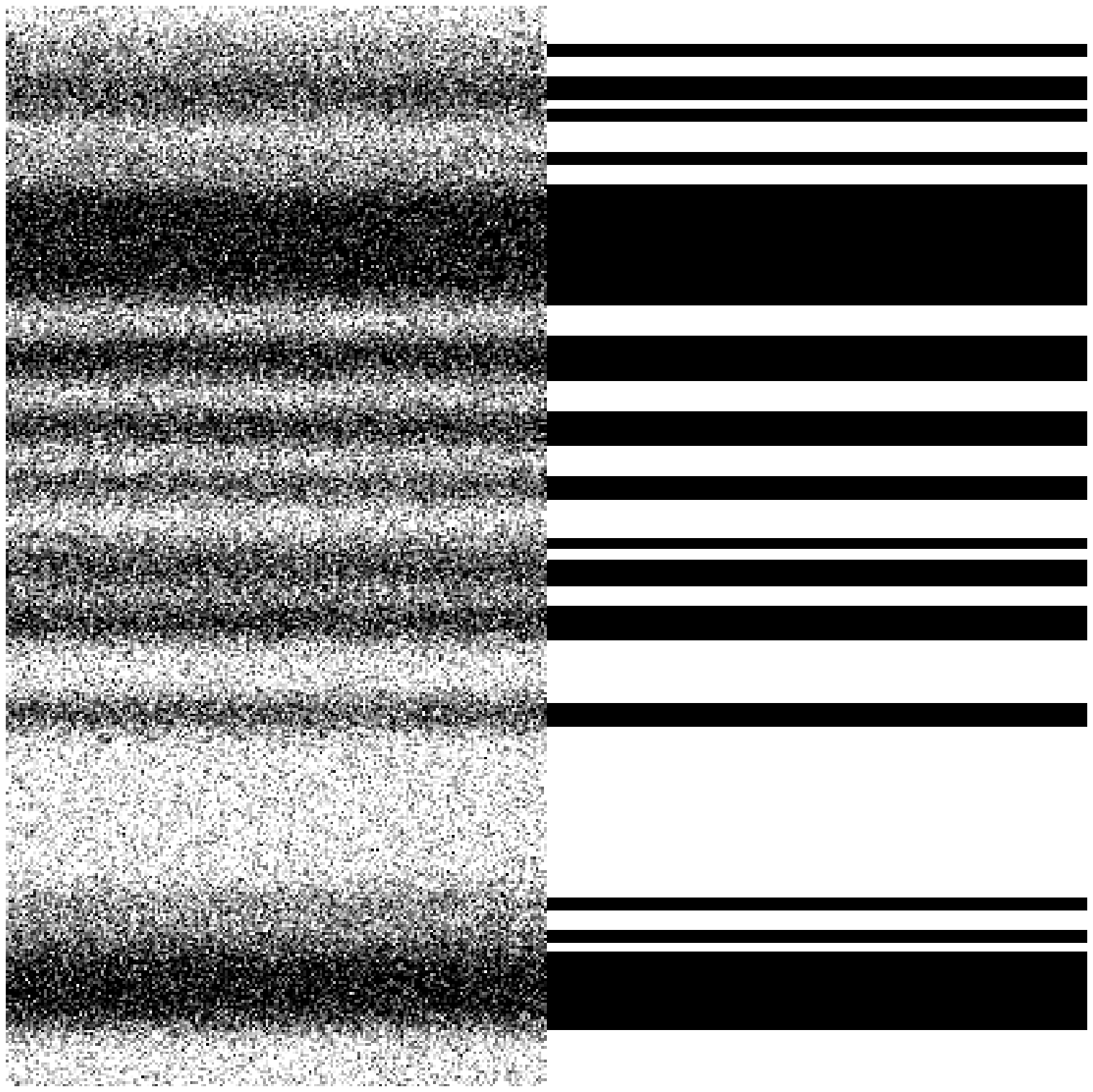}
\end{center}
\caption{The second image blurred with a Gaussian filter with $\sigma=5$, generated by the Matlab command \texttt{fspecial}. The amplitude of Gaussian noise is respectively $0.1$, $0.2$, $0.3$, generated by the Matlab command \texttt{randn}. The noisy image is on the left, the reconstruction is on the right. The average computational time for the reconstructions is respectively 0.04s, 0.03s, 0.03s.}\label{Fig7}
\end{figure}

\begin{figure}[]
\begin{center}
\includegraphics[width=1in]{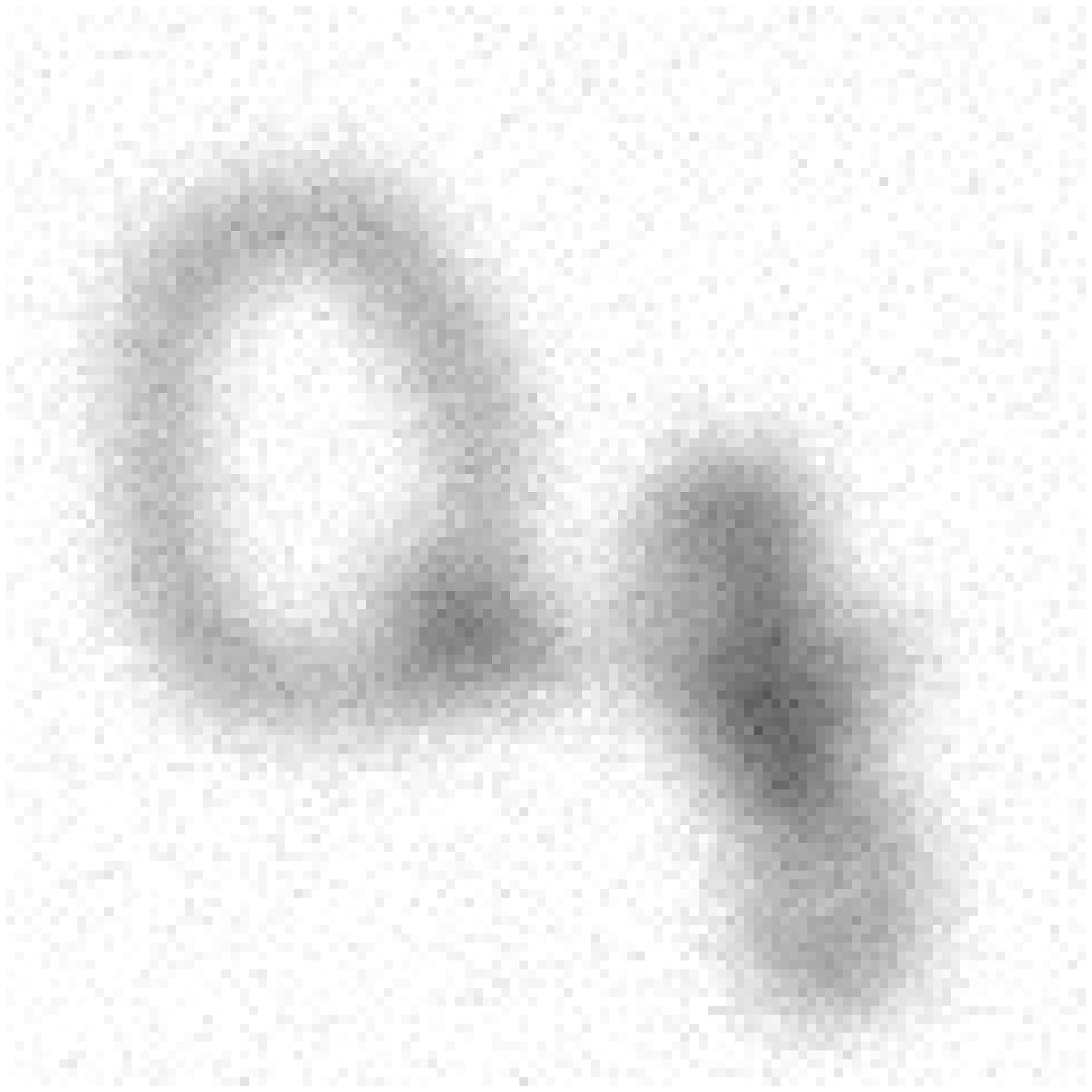}
\includegraphics[width=1in]{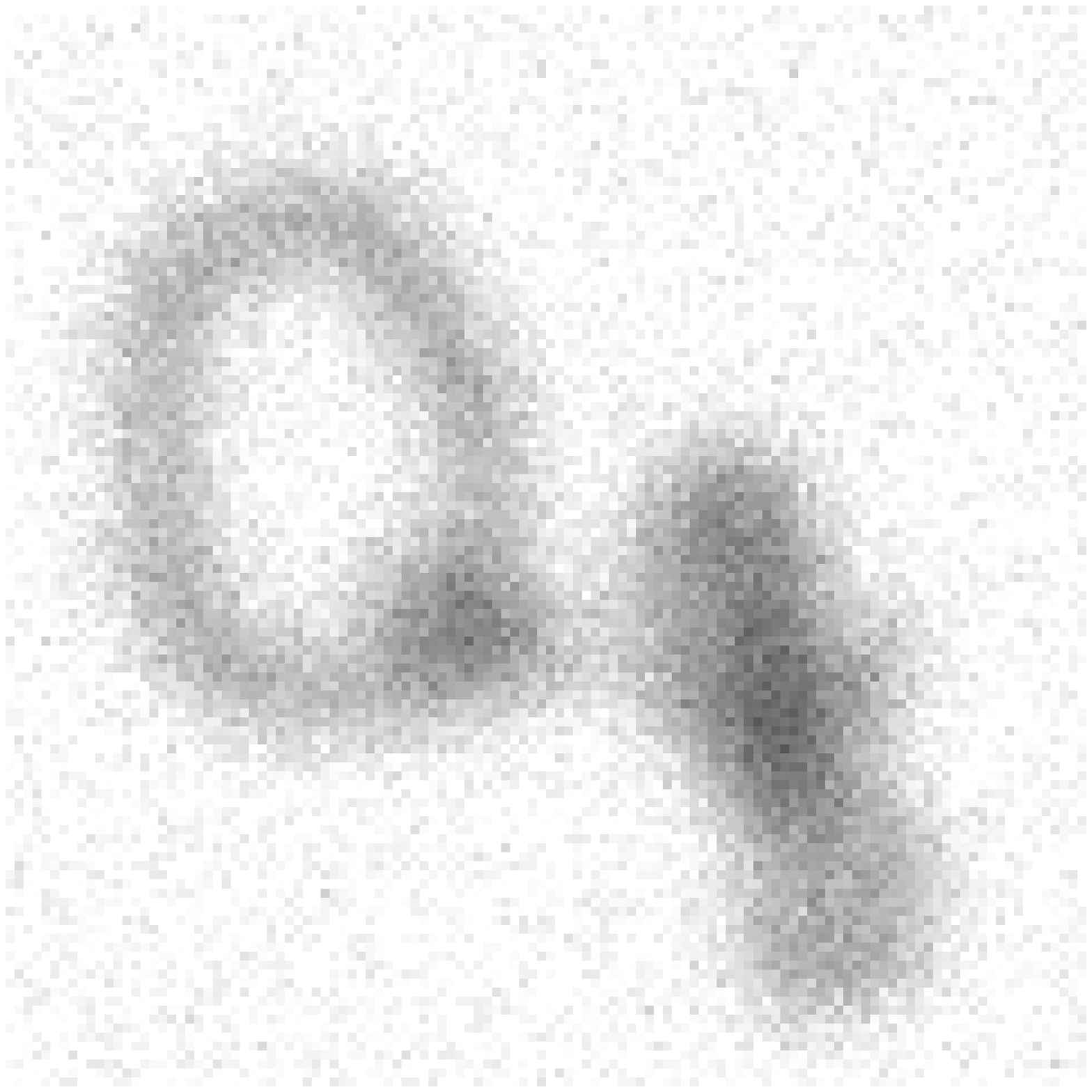}
\includegraphics[width=1in]{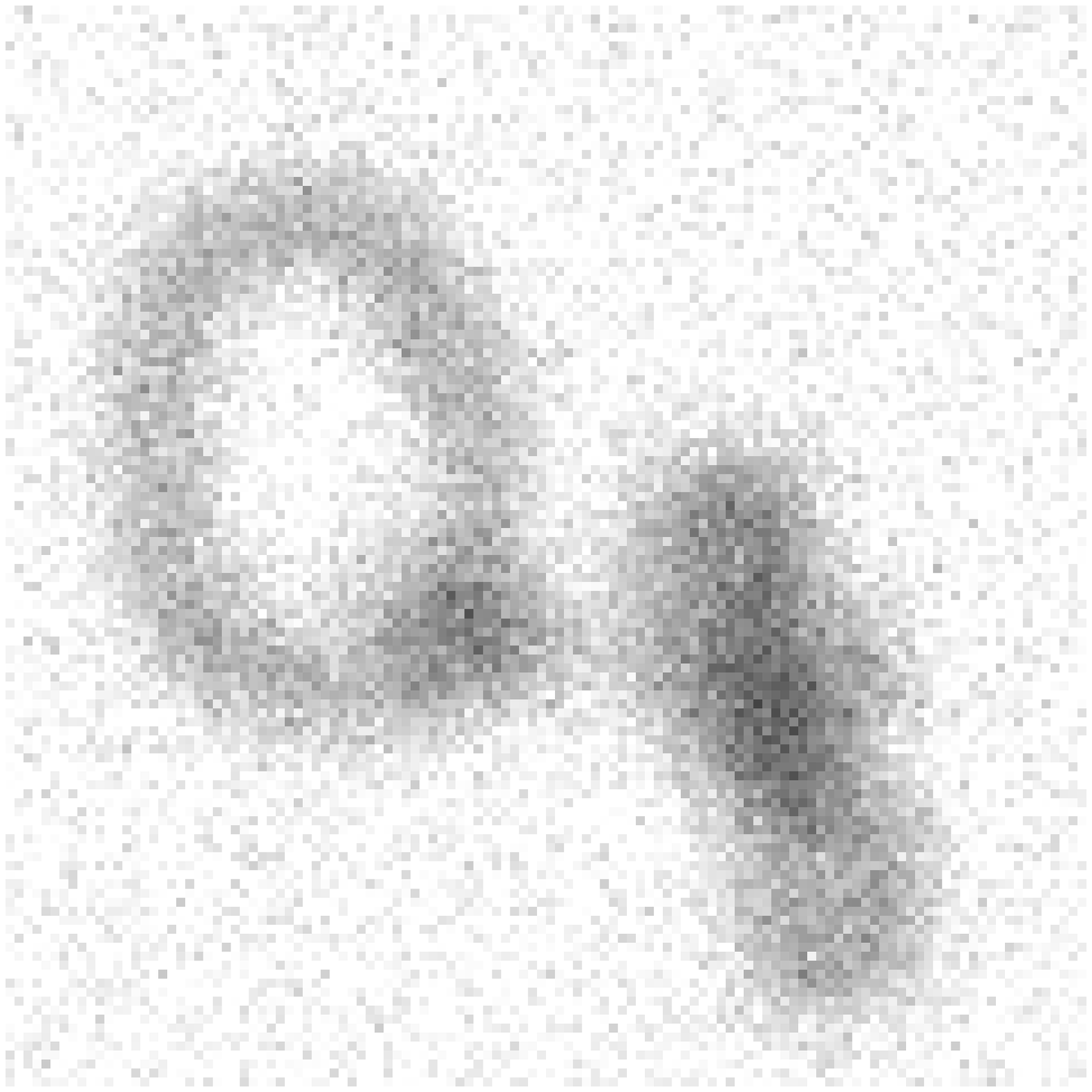}\\
\includegraphics[width=1in]{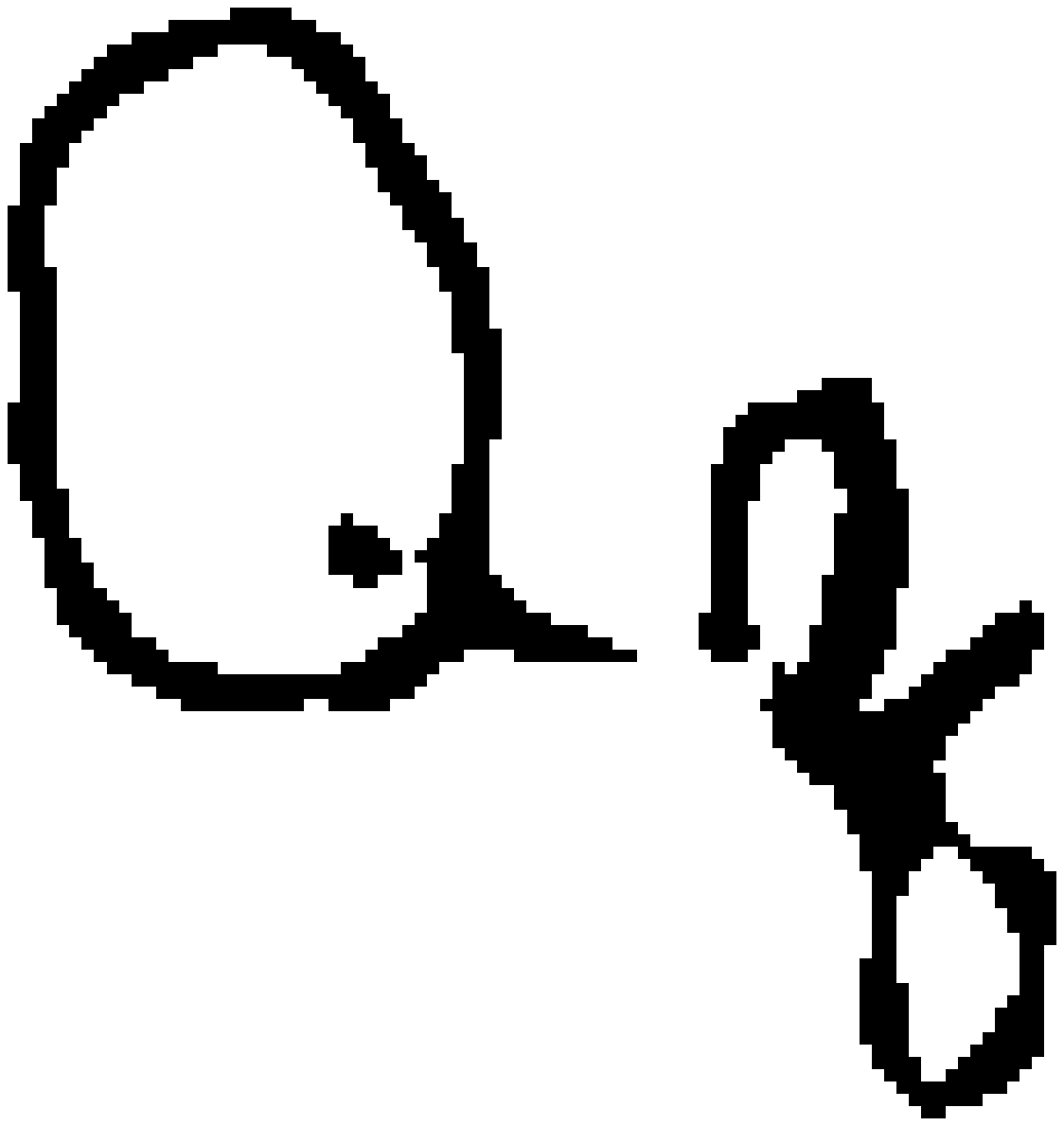}
\includegraphics[width=1in]{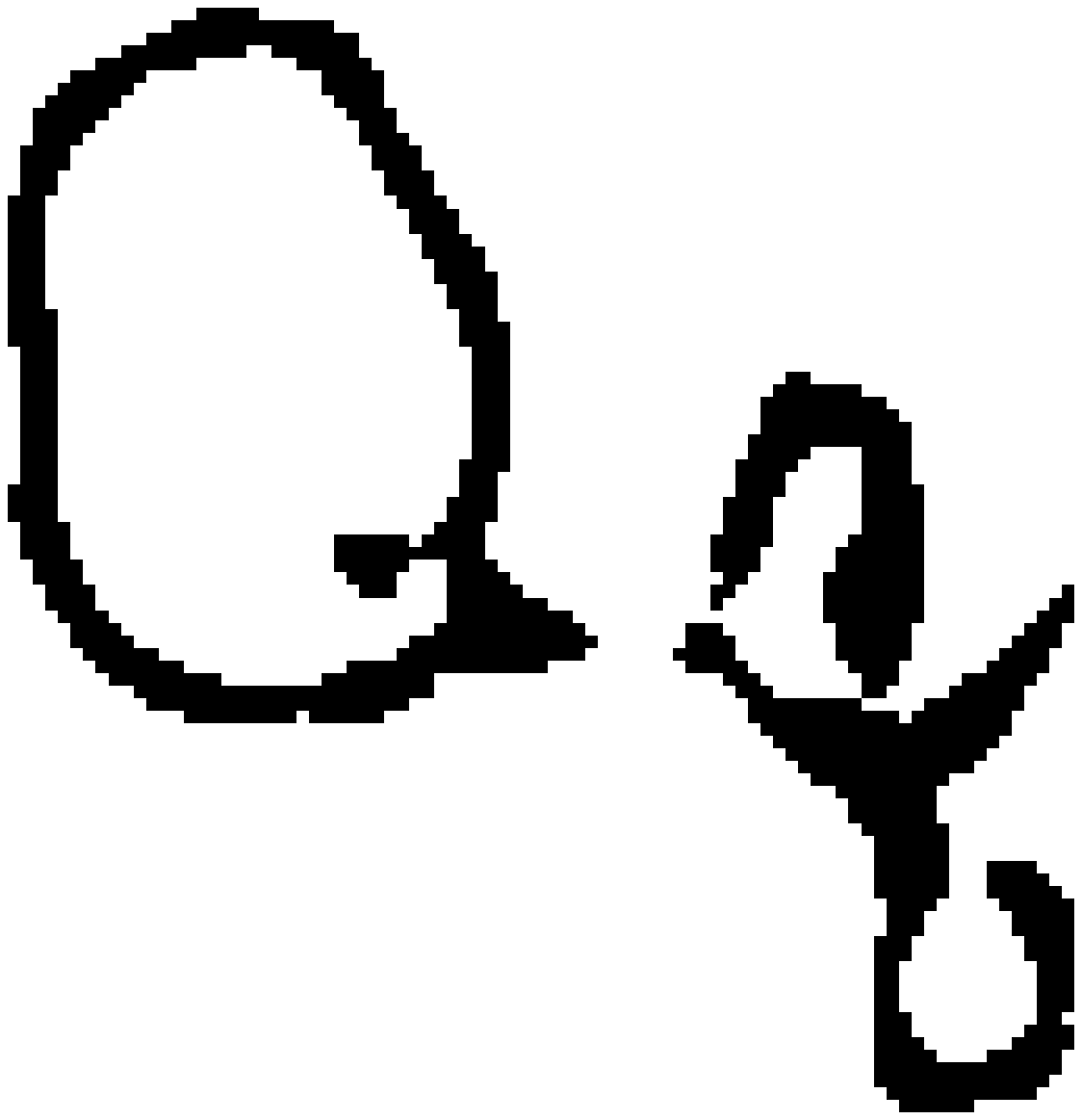}
\includegraphics[width=1in]{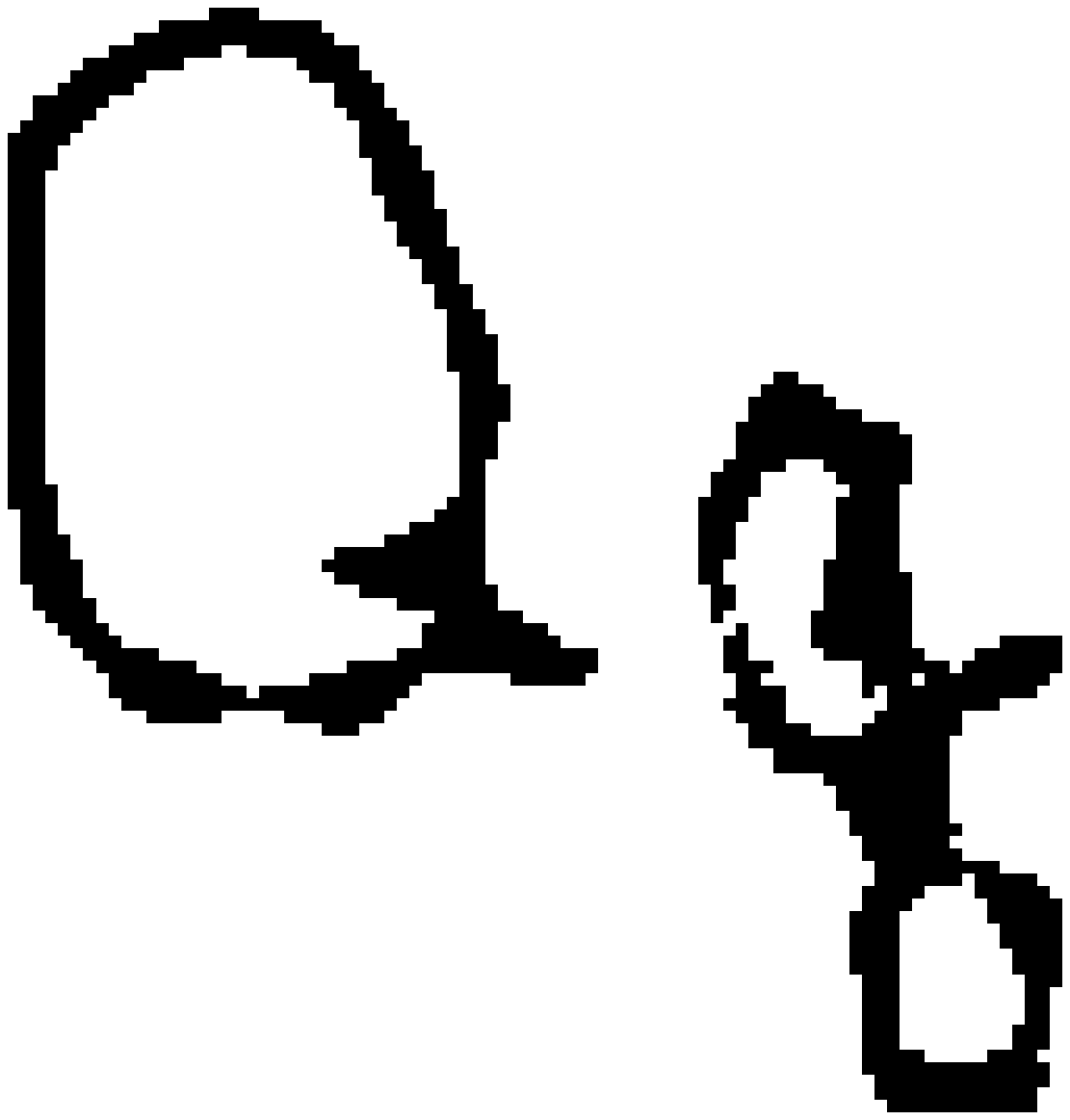}
\end{center}
\caption{The last image blurred with a Gaussian filter with $\sigma=5$, generated by the Matlab command \texttt{fspecial}. The amplitude of the Gaussian noise is respectively $0.03$, $0.05$, $0.07$, generated by the Matlab command \texttt{randn}. The average computational time for the reconstructions is respectively 1.54s, 1.20s, 0.98s.}\label{Fig8}
\end{figure}

\section{Conclusion}\label{conclusion}

The reconstruction of binary functions is difficult because of the nonconvex nature of the problem. In this work we proved that even with very few frequency measurements, binary functions can actually be reconstructed by solving a very simple convex problem. We also discussed a numerical implementation of a solver for this type of convex problem. 

There are several directions for further research. Some of them have been discussed in section \ref{discussion} and \ref{theory2D}. Other potential questions include: How can we investigate more properties of a given binary function by using fewer measurements? Is that possible to even characterize the motion and evolution of a shape by this method? 

\section{Acknowledgement}

The research is supported by the Institute for Mathematics and Its Applications in University of Minnesota. The author wishes to thank Prof. Fadil Santosa from University of Minnesota, Prof. Selim Esedoglu from University of Michigan, Prof. Stanley Osher from UCLA, Prof. Bin Dong from University of Arizona, Prof. Arthur Szlam from the City College of New York and Dr. Jianfeng Lu from the Courant Institute for helpful discussions.

\appendix

\section{Proofs}

\begin{theorem*}\textbf{\emph{\ref{theorem1}.}}
Assume $u_0$ is a binary solution of $Au_0=b$. There exists no nonzero $v\in\{Av=0\}$ such that 
\[
\begin{cases}v(x)\leq0, &\text{when}\quad u_0(x)=1\\
v(x)\geq0, &\text{when}\quad u_0(x)=0
\end{cases}
\]
if and only if $u_0$ is the unique solution of $(P_1)$, i.e. solving $(P_1)$ recovers $u_0$.
\end{theorem*}
\begin{proof} If $(P_1)$ has another solution other than $u_0$, denoted as $u'$, then let $v=u'-u_0$, we have $Av=Au_0-Au'=0$. Moreover, since $u'(x)\in[0,1], \forall x$, then when $u_0(x)=0$, $v(x)=u'(x)-u_0(x)\geq0$; when $u_0(x)=1$, $v(x)=u'(x)-u_0(x)\leq0$, which contradicts the given condition on $v$. The other direction is similar. 
\end{proof}

\begin{lemma*}\textbf{\emph{\ref{trigonointerp}.}}
Let $\T=[0,1]$ where $0$ and $1$ are identified, i.e. $\T\cong S^1=\{z:|z|=1\}$. Given $2n$ points on $\T$ who define $2n$ intervals on $\T$, there exists a real trigonometric polynomial, whose spectrum is limited in $[-n,n]$, vanishing only at those points and changes signs alternatively on those intervals.
\end{lemma*}
\begin{proof} This conclusion is natural in the context of trigonometric interpolation. We denote the $2n$ points by
\[\{c_k=e^{2\pi i\alpha_k}\}_{k=1}^{2n}\subset S^1=\{z:|z|=1\}.
\] 
Let 
\[
u(z)=\frac{C}{z^n}\prod_{k=1}^{2n}(z-c_k)
\]
where $C=\prod_{k=1}^{2n}c_k^{-1/2}$. 
Then $u(z)$ can be written as 
\[u(z)=\sum_{k=-n}^n a_kz^k,
\] i.e. $u$ is a trigonometric polynomial with spectrum limited in $[-n,n]$. Since $\{c_i\}\subset S^1$, $C\in S^1$, for $z\in S^1$ we have 
\[
\overline{u(z)}=\frac{\bar C} {\bar z^n}\prod_{k=1}^{2n}(\bar z-\bar c_k)=\frac{C}{z^n}\prod_{k=1}^{2n}(z-c_k)=u(z)
\]
therefore $u(z)$ is real on $S^1$. If we treat $u(z)=u(e^{2\pi it})$ as a function defined on $[0,1]$, it is easy to check that $u$ vanishes on and only on $\{\alpha_k\}$ and $\frac{d}{dt}u(e^{2\pi it})$ does not vanish on $\{\alpha_k\}$, and the conclusion follows.
\end{proof}

\begin{theorem*}\textbf{\emph{\ref{lowfreq1D}.}}
If $u_0(x)$ is a 1-D binary signal that can be represented as in \eqref{blockwise} with $2d$ consecutive intervals of ones and zeros, then by knowing the Fourier coefficients $\{a_k\}$ for $|k|\leq d$, we can recover $u_0$ through the convex problem $(P_1)$. (Notice that $u_0(x)\in\R, \forall x$ implies $a_k=\overline{a_{-k}}, \forall k$, so essentially we only need to know $\{a_k\}$ for $0\leq k\leq d$.) This result is optimal, i.e. precise reconstruction via solving $(P_1)$ is impossible if knowing even less.
\end{theorem*}
\begin{proof}
By theorem \ref{theorem1p}, we only need to construct a discrete signal $v=A^\top\eta=\F^{-1}(S\eta)$ that changes sign only at the endpoints of the intervals, where $S$ is the sampling operator selecting $\{a_k:|k|\leq d\}$. Let the starting points of each interval be $\{s_i\}$, by lemma \ref{trigonointerp} we can find a trigonometric polynomial $\sum_{k=-d}^{d}a_ke^{2\pi ikt}$ that changes sign only at $\left\{\frac{1}{N}\left(s_i-\frac{1}{2}\right)\right\}.$ Let 
\[v(x)=\sum_{k=-d}^{d}a_ke^{2\pi ik\frac{x}{N}},\quad x\in\{1,\ldots,N\},
\] then $v$ (or $-v$) satisfies the requirement. To show that this result is optimal, we only need to notice that a trigonometric polynomial with order less than $d$ cannot have $2d$ zeros due to the fundamental algebraic theorem. 
\end{proof}

\begin{theorem*}\textbf{\emph{\ref{levelsetcondition}.}}
Assume $u(x,y)$ is a 2D binary function with analytic jump curve and the average directional number of zero-crossings of $u$ along the angle $\theta$ is denoted as  $K_\theta$. If there exists a band-limited real function $v(x,y)=\sum_{(j,k)\in\Omega} a_{jk}e^{2\pi i(jx+ky)}$ defined on $\T^2$, where $\Omega=\{(j,k):\sqrt{j^2+k^2}\leq d\}$, such that the jump set of $u(x,y)$ corresponds to the zero levelset of $v(x,y)$, then  $K_\theta\leq 2d,\forall \theta$.
\end{theorem*}
\begin{proof}
In this proof we are discussing the problem on a torus $\T^2$, so every coordinates are automatically mod by $1$ without explicitly written to make the notations clear. 

Without loss of generality, we only prove the case when $\theta\in(-\pi/4,\pi/4]$. The other case can be automatically proved by switching the $x$ and $y$ coordinates. 

Since the zero levelset of $v(x,y)$ are analytic curves, $K_{\theta}(s)$ is a piece-wise constant function with respect to both $s$ and $\theta$ with finitely many jumps, and it is easy to see that we can only prove the theorem for $\theta$ in $\{\theta:\tan\theta\in \Q,\theta\neq 0,1\}$ because this set is dense in $(-\pi/4,\pi/4]$.

For any $\theta\neq 0$ s.t. $\tan\theta\in \Q$, assume 
\[
\tan\theta=p/q,\quad p<q\in\Z\text{ are coprime.}
\]Let 
\[
F_{s,\theta}(t)=(t,s+t\tan\theta)\mod 1,\quad t\in\R 
\]
be a linear flow on $\T^2$, then $F_{s,\theta}(t)$ is a periodic function with period $q$ since $F_{s,\theta}(t+q)\equiv (t+q,s+(t+q)\tan\theta)\equiv (t,s+t\tan\theta)\equiv F_{s,\theta}(t)\mod 1$. 

The main idea of the proof is based on the following observation: a whole period of  $F_{s,\theta}$ can be split to $q$ segments $L_{s_m,\theta}$, $m=0,1,2,\ldots, q-1$. Therefore, counting the average intersection along the line segments $L_{s,\theta}$ can be replaced by counting the intersection along $F$. The latter is easier because it can be reduced to counting the 1D zero crossings of a trigonometric polynomial inside a period. 

Recall that
\[
L_{s,\theta}(t)=(t, s+t\tan\theta)\mod 1, \quad s\in[0,1], t\in[0,1]
\]
We will first show that $F_{s,\theta}(t),\quad t\in[0,q]$ can be split to $q$ segments $L_{s_m,\theta}$, $m=0,1,2,\ldots, q-1$. Indeed, when $t\in[m,m+1]$ for $m=0,1,2,\ldots, q-1$, it is easy to check 
\[F_{s,\theta}(t)=L_{s+mp/q,\theta}(t-m),
\] so when $t$ goes from $0$ to $q$, $F_{s,\theta}(t)$ can be seen as the connected version of $q$ line segments $\{L_{s+mp/q,\theta}: m=0,\ldots,q-1\}$. Since $p$ and $q$ are coprime, elementary number theory tells us that 
\[
\left\{\text{mod}\left(mp/q,1\right)\right\}_{m=0,1,2,\ldots, q-1}= \left\{\text{mod}\left(m/q,1\right)\right\}_{m=0,1,2,\ldots, q-1}.
\]
Similar to the notation of $\#L_{s,\theta}$, we use $\#F_{s,\theta}$ to denote the intersection of $F_{s,\theta}$, $t\in[0,q]$ with the jump set of $u_0$, then 
\[
\#F_{s,\theta}=\sum_{m=0}^{q-1}\#L_{s+mp/q,\theta}=\sum_{m=0}^{q-1}\#L_{s+m/q,\theta}
\]
Therefore, we have
\begin{align}\nonumber
K_\theta&=\cos\theta\int_0^1 \#L_{s,\theta}ds=\cos\theta\sum_{m=0}^{q-1}\int_{\frac{m}{q}}^{\frac{m+1}{q}}\#L_{s,\theta}ds\\ \nonumber
&= \cos\theta\sum_{m=0}^{q-1}\int_0^{\frac{1}{q}}\#L_{s+m/q,\theta}ds\\ \nonumber
&=\cos\theta\int_0^{\frac{1}{q}}\left(\sum_{m=0}^{q-1}\#L_{s+m/q,\theta}\right)ds\\
&=\cos\theta\int_0^{\frac{1}{q}}\#F_{s,\theta}ds \label{K_theta_F}
\end{align}
That is to say, we replace $\#L_{s,\theta}$ in the definition of $K_\theta$ by $\#F_{s,\theta}$. Now we start to evaluate $\#F_{s,\theta}$. Since  $u_0(x,y)$ is corresponding to the zero levelset of $v(x,y)$, $\#F_{s,\theta}$ is equal to the number of zero-crossings of $v(x,y)$ along $F_{s,\theta}$, and we need to count the zero-crossings of $v(F_{s,\theta}(t))$ when $t\in[0,q]$. Let $\tilde v(t)=v(F_{s,\theta}(t))$, recall 
\[v(x,y)=\sum_{(j,k)\in\Omega} a_{jk}e^{2\pi i(jx+ky)},
\] we have
\[
\tilde v(t)=\sum_{(j,k)\in\Omega} a_{jk}e^{2\pi i(jt+k(s+tp/q))}
\]so
\be\label{vqt}
\tilde v(qt)=\sum_{(j,k)\in\Omega} a_{jk}e^{2\pi iks}e^{2\pi i(jq+kp)t}.
\ee
Since $F_{s,\theta}(t)$ is periodic with period $q$, so is $\tilde v(t)$, thus $\tilde v(qt)$ as a function of $t$ is periodic with period $1$. Since $\sqrt{j^2+k^2}\leq d$ in $\Omega$, we have
\[
|jq+kp|=\sqrt{(j^2+k^2)(p^2+q^2)-(jp-kq)^2}\leq d\sqrt{p^2+q^2}
\]
Then \eqref{vqt} tells us that $\tilde v(qt)$ as a function of $t$ can be expanded as a trigonometric polynomial with order no more than $d\sqrt{p^2+q^2}$, therefore the zero-crossing of $\tilde v(qt)$ is no more than $2 d\sqrt{p^2+q^2}$, i.e. $\#F_{s,\theta}\leq 2 d\sqrt{p^2+q^2}$. Plug it into \eqref{K_theta_F}, we have $K_\theta\leq2d\cos\theta\sqrt{p^2+q^2}/q=2d$.

\end{proof}

\begin{theorem*}\textbf{\emph{\ref{robustness}.}}
If $u_0$ is the unique solution of $(P_1)$, $\tilde b=b+ \epsilon $ is the corrupted measurement,  $\tilde u$ is the minimizer of the optimization problem \[
\min_u \|Au-\tilde b\|_2^2\quad\st\quad 0\leq u\leq 1,\]
 then when $\|\epsilon\|< h(A, \mathbb{O}_{u_0})$ where $h>0$ is a small amount depending only on $A$ and $\mathbb{O}_{u_0}$ (see details in the proof), $B(\tilde u)=u_0$. 
\end{theorem*}
\begin{proof}
\noindent We introduce the following lemma: \\
\\
\textbf{Lemma:} \emph{Let $\mathbb{O}$ be an orthant such that $\mathbb{O}\cap\ker A=\{0\}$. When $\|\epsilon\|< h(A,\mathbb{O})$ where $h>0$ is a small amount depending only on $A$ and $\mathbb{O}$, the solution $w^*$ of the linear programming problem 
\[\min_w \|Aw-\epsilon\|\quad \st\quad w\in\mathbb{O}
\]satisfies $|w^*_i|< 1/2,\forall i$.}

At first we demonstrate that this lemma implies the conclusion we need. Because $u_0$ is the unique solution of $(P_1)$, theorem \ref{theorem1} tells us that $\mathbb{O}_{u_0}\cap\ker A=\{0\}$. If $\tilde u$ solves 
\[\min_u \|Au-\tilde b\|_2^2\quad \st\quad0\leq u\leq 1,
\]
then because $\tilde b=b+\epsilon=Au_0+\epsilon$, it is easy to see that $w^*=\tilde u-u_0$ solves 
\[\min_w \|Aw-\epsilon\|\quad \st\quad w\in\mathbb{O}_{u_0}
\]
From the above lemma, since $\|\epsilon\|$ is small enough, we have $|w^*_i|<1/2,\forall i$. Therefore $B(\tilde u)=u_0$ and the conclusion follows.

Now we go back to prove the lemma. 

Consider the linear programming problem
\[\min_w \|Aw-\epsilon\|\quad \st\quad w\in\mathbb{O}
\]
The duality problem is:
\[\max_{\mu}\, -\mu^\top\epsilon \quad \st\quad A^\top\mu\in \mathbb{O},\,\|\mu\|\leq 1\]
The Karush-Kuhn-Tucker conditions of the optimal variables of the primal and duality problems are
\be
w^*\in\mathbb{O},\,A^\top\mu^*\in \mathbb{O},\|\mu^*\|\leq 1
\ee
\be\label{compslack1}
w^*_i\cdot (A^\top\mu^*)_i=0,\forall i
\ee
\be\label{compslack2}
\mu=\frac{Aw^*-\epsilon}{\|Aw^*-\epsilon\|} \text{ if } Aw^*-\epsilon\neq 0
\ee
where the latter two are the complementary slackness conditions. From \eqref{compslack1} we see that $\langle w^*,A^\top \mu^*\rangle=0$, so when $Aw^*-\epsilon\neq 0$, by \eqref{compslack2} we get
\begin{align}\nonumber
\langle Aw^*,Aw^*-\epsilon \rangle&=\|Aw^*-\epsilon\|\langle Aw^*,\mu \rangle\\\label{orthogonality}
&=\|Aw^*-\epsilon\|\langle w^*,A^\top \mu^*\rangle=0.
\end{align}
Apparently \eqref{orthogonality} also holds when $Aw^*-\epsilon=0$, so it is satisfied anyway. Therefore, Pythagorean's theorem shows \[\|Aw^*\|^2=\|\epsilon\|^2-\|Aw^*-\epsilon\|^2\leq\|\epsilon\|^2,\] i.e. $\|Aw^*\|\leq \|\epsilon\|$.

By the alternative theorem \ref{alternative}, from $\mathbb{O}\cap\ker A=\{0\}$ we know that there exists $v=A^\top \eta \in\text{int}(\mathbb{O})$. Without loss of generality we assume $\|\eta\|=1$. Let $h$ be a positive number such that $|v_i|=|(A^\top\eta)_i|\geq 2h,\forall i$, then $h$ depends on $A$ and $\mathbb{O}$ only. When $\|\epsilon\|<h$ we have
\begin{align*}
\sum_iw^*_iv_i&=\langle w^*,v\rangle=\langle w^*,A^\top\eta\rangle\\
&=\langle Aw^*,\eta\rangle\leq \|Aw^*\|\|h\|=\|Aw^*\|\leq\|\epsilon\|< h.
\end{align*}
Because $w^*\in\mathbb{O}$ and $v\in\text{int}(\mathbb{O})$, $w^*_iv_i\geq 0,\forall i$. Therefore, $w^*_iv_i\leq\sum_iw^*_iv_i< h,\forall i$. From $|v_i|\geq 2h,\forall i$ we have $|w^*_i|< h/|v_i|\leq 1/2,\forall i$. That finishes the proof of the lemma.
\end{proof}

\IEEEpeerreviewmaketitle

\bibliographystyle{IEEEtranS}
\bibliography{IEEEabrv,biblist}

\begin{thebibliography}{10}
\providecommand{\url}[1]{#1}
\csname url@samestyle\endcsname
\providecommand{\newblock}{\relax}
\providecommand{\bibinfo}[2]{#2}
\providecommand{\BIBentrySTDinterwordspacing}{\spaceskip=0pt\relax}
\providecommand{\BIBentryALTinterwordstretchfactor}{4}
\providecommand{\BIBentryALTinterwordspacing}{\spaceskip=\fontdimen2\font plus
\BIBentryALTinterwordstretchfactor\fontdimen3\font minus
  \fontdimen4\font\relax}
\providecommand{\BIBforeignlanguage}[2]{{%
\expandafter\ifx\csname l@#1\endcsname\relax
\typeout{** WARNING: IEEEtranS.bst: No hyphenation pattern has been}%
\typeout{** loaded for the language `#1'. Using the pattern for}%
\typeout{** the default language instead.}%
\else
\language=\csname l@#1\endcsname
\fi
#2}}
\providecommand{\BIBdecl}{\relax}
\BIBdecl

\bibitem{Alajlan:2008p7326}
N.~Alajlan, M.~Kamel, and G.~Freeman, ``Geometry-based image retrieval in
  binary image databases,'' \emph{Pattern Analysis and Machine Intelligence,
  IEEE Transactions on}, vol.~30, no.~6, pp. 1003 -- 1013, 2008.

\bibitem{Boyd:2004p1827}
S.~Boyd and L.~Vandenberghe, ``Convex optimization,'' 2004.

\bibitem{Bresson:2007p7744}
X.~Bresson, S.~Esedoglu, P.~Vandergheynst, J.-P. Thiran, and S.~Osher, ``Fast
  global minimization of the active contour/snake model,'' \emph{J Math Imaging
  Vis}, vol.~28, no.~2, pp. 151--167, Jan 2007.

\bibitem{Bruckstein:2008p2703}
A.~M. Bruckstein, M.~Elad, and M.~Zibulevsky, ``On the uniqueness of
  nonnegative sparse solutions to underdetermined systems of equations,''
  \emph{IEEE Transactions on Information Theory}, vol.~54, no.~11, pp.
  4813--4820, 2008.

\bibitem{Candes:2006p625}
E.~J. Candes, J.~Romberg, and T.~Tao, ``Robust uncertainty principles: exact
  signal reconstruction from highly incomplete frequency information,''
  \emph{IEEE Transactions on Information Theory}, vol.~52, no.~2, pp. 489--
  509, 2006.

\bibitem{Candes:2006p1807}
------, ``Stable signal recovery from incomplete and inaccurate measurements,''
  \emph{Communications on Pure and Applied Mathematics}, vol.~59, no.~8, 2006.

\bibitem{Chan:2006p2646}
T.~F. Chan, S.~Esedoglu, and M.~Nikolova, ``Algorithms for finding global
  minimizers of image segmentation and denoising models,'' \emph{Siam J Appl
  Math}, vol.~66, no.~5, pp. 1632--1648, Jan 2006.

\bibitem{Cover:1965p7055}
T.~M. Cover, ``Geometrical and statistical properties of systems of linear
  inequalities with applications in pattern recognition,'' \emph{Electronic
  Computers, IEEE Transactions on}, vol. EC-14, no.~3, pp. 326 -- 334, 1965.

\bibitem{CURTIS:1987p7208}
S.~Curtis and A.~Oppenheim, ``Reconstruction of multidimensional signals from
  zero crossings,'' \emph{J Opt Soc Am A}, vol.~4, no.~1, pp. 221--231, Jan
  1987.

\bibitem{Curtis:1985p7206}
S.~Curtis, A.~Oppenheim, and J.~Lim;, ``Signal reconstruction from fourier
  transform sign information,'' \emph{Acoustics, Speech and Signal Processing,
  IEEE Transactions on}, vol.~33, no.~3, pp. 643 -- 657, 1985.

\bibitem{Donoho:2010p5288}
D.~Donoho and J.~Tanner, ``Exponential bounds implying construction of
  compressed sensing matrices, error-correcting codes, and neighborly polytopes
  by random sampling,'' \emph{Information Theory, IEEE Transactions on},
  vol.~56, no.~4, pp. 2002 -- 2016, 2010.

\bibitem{Donoho:2006p668}
D.~L. Donoho, ``Compressed sensing,'' \emph{IEEE Transactions on Information
  Theory}, vol.~52, no.~4, pp. 1289--1306, 2006.

\bibitem{Donoho:1989p3875}
D.~L. Donoho and P.~Stark, ``Uncertainty principles and signal recovery,''
  \emph{Siam J Appl Math}, vol.~49, no.~3, pp. 906--931, Jun 1989.

\bibitem{Donoho:2010p6951}
D.~L. Donoho and J.~Tanner, ``Counting the faces of randomly-projected
  hypercubes and orthants, with applications,'' \emph{Discrete Comput Geom},
  vol.~43, no.~3, pp. 522--541, Jan 2010.

\bibitem{Esedoglu:2004p7220}
S.~Esedoglu, ``Blind deconvolution of bar code signals,'' \emph{Inverse
  Problems}, vol.~20, no.~1, pp. 121--135, Jan 2004.

\bibitem{Fuchs:2005p8032}
J.~Fuchs, ``Sparsity and uniqueness for some specific under-determined linear
  systems,'' \emph{Acoustics, Speech, and Signal Processing, 2005. Proceedings.
  (ICASSP '05). IEEE International Conference on}, vol.~5, pp. v/729 -- v/732
  Vol. 5, 2005.

\bibitem{Galaktionov:2005p7108}
V.~A. Galaktionov and P.~J. Harwin, ``Sturm's theorems on zero sets in
  nonlinear parabolic equations,'' pp. 173--199, 2005.

\bibitem{Goldstein:2009p2473}
T.~Goldstein and S.~Osher, ``The split bregman method for l1 regularized
  problems,'' \emph{SIAM J. Imaging Sci}, vol.~2, no.~2, pp. 323--343, 2009.

\bibitem{Hummel:1989p7209}
R.~Hummel and R.~Moniot, ``Reconstructions from zero crossings in scale
  space,'' \emph{Acoustics, Speech and Signal Processing, IEEE Transactions
  on}, vol.~37, no.~12, pp. 2111 -- 2130, 1989.

\bibitem{Ishikawa:2003p7731}
H.~Ishikawa, ``Exact optimization for markov random fields with convex
  priors,'' \emph{Pattern Analysis and Machine Intelligence, IEEE Transactions
  on}, vol.~25, no.~10, pp. 1333 -- 1336, 2003.

\bibitem{Kedem:1986p7053}
B.~Kedem, ``Spectral analysis and discrimination by zero-crossings,''
  \emph{Proceedings of the IEEE}, vol.~74, no.~11, pp. 1477 -- 1493, 1986.

\bibitem{Kontogiorgis:1998p8034}
S.~Kontogiorgis and R.~Meyer, ``A variable-penalty alternating directions
  method for convex optimization,'' \emph{Math. Programming}, vol.~83, no.~1,
  pp. 29--53, Jan 1998.

\bibitem{Kozma:2002p7049}
G.~Kozma and F.~Oravecz, ``On the gaps between zeros of trigonometric
  polynomials,'' \emph{Real Anal. Exchange}, vol.~28, no.~2, pp. 447--454,
  2002.

\bibitem{Litman:1998p7217}
A.~Litman, D.~Lesselier, and F.~Santosa, ``Reconstruction of a two-dimensional
  binary obstacle by controlled evolution of a level-set,'' \emph{Inverse
  Problems}, vol.~14, no.~3, pp. 685--706, Jan 1998.

\bibitem{Logan:1977p7105}
B.~Logan, ``information in the zero crossings of bandpass signals,'' \emph{Bell
  Syst Tech J}, vol.~56, no.~4, pp. 487--510, Jan 1977.

\bibitem{Nashold:1989p7110}
K.~M. Nashold, J.~A. Bucklew, W.~Rudin, and B.~E.~A. Saleh, ``Synthesis of
  binary images from band-limited functions,'' \emph{J. Opt. Soc. Amer. A},
  vol.~6, no.~6, pp. 852--858, 1989.

\bibitem{Osher:2005p632}
S.~Osher, M.~Burger, D.~Goldfarb, J.~Xu, and W.~Yin, ``An iterative
  regularization method for total variation-based image restoration,''
  \emph{Multiscale Model Sim}, vol.~4, no.~2, pp. 460--489, Jan 2005.

\bibitem{Osher:2001p7372}
S.~Osher and F.~Santosa, ``Level set methods for optimization problems
  involving geometry and constraints i. frequencies of a two-density
  inhomogeneous drum,'' \emph{J Comput Phys}, vol. 171, no.~1, pp. 272--288,
  Jan 2001.

\bibitem{Pock:2010p7690}
T.~Pock, D.~Cremers, and H.~Bischof{\ldots}, ``Global solutions of variational
  models with convex regularization,'' \emph{SIAM J. IMAGING SCIENCES}, Jan
  2010.

\bibitem{Requicha:1980p7107}
A.~Requicha, ``The zeros of entire functions: Theory and engineering
  applications,'' \emph{Proceedings of the IEEE}, vol.~68, no.~3, pp. 308 --
  328, 1980.

\bibitem{Rotem:1986p7178}
D.~Rotem and Y.~Zeevi, ``Image reconstruction from zero crossings,''
  \emph{Acoustics, Speech and Signal Processing, IEEE Transactions on},
  vol.~34, no.~5, pp. 1269 -- 1277, 1986.

\bibitem{Sanz:1989p7201}
J.~Sanz, ``Multidimensional signal representation by zero crossings: An
  algebraic study,'' \emph{SIAM Journal on Applied Mathematics}, vol.~49,
  no.~1, pp. 281--295, Feb 1989.

\bibitem{Sanz:1989p7122}
J.~Sanz and T.~Huang, ``Image representation by sign information,''
  \emph{Pattern Analysis and Machine Intelligence, IEEE Transactions on},
  vol.~11, no.~7, pp. 729 -- 738, 1989.

\bibitem{Szlam:2010p7216}
A.~Szlam, Z.~Guo, and S.~Osher, ``A split bregman method for non-negative
  sparsity penalized least squares with applications to hyperspectral
  demixing,'' \emph{Image Processing (ICIP), 2010 17th IEEE International
  Conference on}, pp. 1917 -- 1920, 2010.

\bibitem{Ulanovskii:2006p7050}
A.~Ulanovskii, ``The sturm-hurwitz theorem and its extensions,'' \emph{J
  Fourier Anal Appl}, vol.~12, no.~6, pp. 629--643, Jan 2006.

\bibitem{Vetterli:2002p7065}
M.~Vetterli, P.~Marziliano, and T.~Blu, ``Sampling signals with finite rate of
  innovation,'' \emph{Signal Processing, IEEE Transactions on}, vol.~50, no.~6,
  pp. 1417 -- 1428, 2002.

\bibitem{Zakhor:1990p7196}
A.~Zakhor and A.~Oppenheim, ``Reconstruction of two-dimensional signals from
  level crossings,'' \emph{Proceedings of the IEEE}, vol.~78, no.~1, pp. 31 --
  55, 1990.

\end{thebibliography}

\end{document}